\documentclass[10 pt, journal]{IEEEtran}

\usepackage{dblfloatfix}

\usepackage{graphicx, subfigure}
\usepackage{amsmath}
\usepackage{amssymb}

\newcommand{\med}{\;|\;}

\newtheorem{theorem}{Theorem}
\newtheorem{corollary}{Corollary}
\newtheorem{lemma}{Lemma}
\newtheorem{remark}{Remark}

\def \bX{\boldsymbol{X}}
\def \bbH{\boldsymbol{H}}
\def \bS{\boldsymbol{S}}

\def \H{{\cal H}}
\def \D{{\cal D}}
\def \J{{\cal J}}
\def \G{{\cal G}}
\def \L{{\cal L}}

\def \bW{\boldsymbol{W}}

\newcommand{\bbe}{\mathbb E}

\def \ber{\mbox{Bernoulli}}
\def \dGamma{\mbox{Gamma\;}}

\newcommand\dotlt{\;{\overset \cdot \leq}\;}
\newcommand{\dff}{\stackrel{\scriptscriptstyle\triangle}{=}}
\newcounter{MYtempeqncnt}

\begin{document}

\title{\huge Adaptive Sensing of Congested Spectrum Bands}

\author{Ali~Tajer,~\IEEEmembership{Member,~IEEE,}~Rui~M.~Castro,~and~Xiaodong~Wang,~\IEEEmembership{Fellow,~IEEE,}
\thanks{A. Tajer is with the Electrical Engineering Department, Princeton University, Princeton, NJ 08544, e-mail:
tajer@princeton.edu}
\thanks{R. M. Castro is with the Department of Mathematics \& Computer Science, Eindhoven University of Technology, Eindhoven, Netherlands, e-mail: rmcastro@tue.nl}
\thanks{X. Wang is with the Electrical Engineering Department, Columbia University, New York, NY 10027, e-mail: wangx@ee.columbia.edu}
\thanks{This work was supported in part by the U.S. National Science Foundation (NSF) under grant CCF-0726480, and by the U.S. Office of Naval Research (ONR) under grant N00014-08-1-0318.}}

\maketitle

\begin{abstract}

Cognitive radios process their sensed information collectively in order to opportunistically identify and access under-utilized spectrum segments (spectrum holes). Due to the transient and rapidly-varying nature of the spectrum occupancy, the cognitive radios (secondary users) must be agile in identifying the spectrum holes in order to enhance their spectral efficiency. We propose a novel {\em adaptive} procedure to reinforce the agility of the secondary users for identifying {\em multiple} spectrum holes simultaneously over a wide spectrum band. This is accomplished by successively {\em exploring} the set of potential spectrum holes and {\em progressively} allocating the sensing resources to the most promising areas of the spectrum. Such exploration and resource allocation results in conservative spending of the sensing resources and translates into very agile spectrum monitoring. The proposed successive and adaptive sensing procedure is in contrast to the more conventional approaches that distribute the sampling resources equally over the entire spectrum. Besides improved agility, the adaptive procedure requires less-stringent constraints on the power of the primary users to guarantee that they remain distinguishable from the environment noise and renders more reliable spectrum hole detection.
\end{abstract}

\section{Introduction}
\label{sec:intro}

\subsection{Background}
The notion of cognitive radio has emerged as a mean to alleviate the scarcity of available frequency spectrum through accommodating unlicensed (secondary) users within the under-utilized spectrum bands licensed to the legacy (primary) users. For this purpose, secondary users monitor the spectrum in order to identify and exploit the under-utilized segments. Such coexistence of the primary and secondary users is constrained by the necessary warranties on the quality of service provided for the primary users.

For detecting an under-utilized segment in the wideband spectrum, which we hereinafter call a spectrum hole, secondary users face two major challenges. First, the spectrum holes are spread across the wide band  and often their occupancy status changes rapidly, i.e., the holes might not remain unoccupied for a long time. Therefore, the amount of time spent searching for a spectrum hole could potentially constitute a considerable fraction of the time that a spectrum hole remains vacant. Such rapid changes are governed mainly by the spectral activities of the active (primary and other potential secondary) users seeking access to the same spectrum. Therefore, the secondary users should be {\em agile} enough to detect a spectrum hole in substantially less time than the period in which the hole remains vacant. Agility becomes even more prominent in wideband systems where the secondary users have to monitor a larger number of channels with a {\em sparse} distribution of spectrum holes. The current statistics about the spectrum occupancy patterns provide that a considerable portion of the spectrum is under-utilized, hence promoting the notion of cognitive communication. Nevertheless, by envisioning the futuristic scenario of cognitive networks, enabling cognitive communication allows many unlicensed users to compete for the same spectrum resources. Under this envisioned scenario and by taking into account the ever-increasing demand for data communication, the vacant spectrum bands will not be as abundant as they are, which makes the sparsity assumption quite reasonable. We would like to also remark that the same futuristic sparsity assumption on spectrum vacancy is also envisioned in \cite{Giannakis:SP10, Han:11}. The second challenge is to avoid harming the communication of the primary users. For this purpose the secondary users must {\em reliably} distinguish the holes from the spectrum segments accessed by the primary users, regardless of how weak the primary users are.

Sensing the spectrum, being a major task of the secondary users, has received a considerable amount of research interest, see, e.g., \cite{Haykin:09, Ma:09,Yucek:09, Akyildiz:06,Lai:IT10_Submitted, Quan:SPM08}. We review some of the related works in the following section.

\subsection{Related Works}
There exist, broadly, three approaches for identifying  spectrum vacancies within a wideband spectrum. A notable approach on agile spectrum sensing in a wideband spectrum is the quickest sequential search scheme of \cite{Lai:IT10_Submitted}. In this scheme the spectrum is split into smaller narrowband channels and the secondary user senses one channel at a time. After accumulating enough information about each channel it decides whether the channel is a hole or is occupied. If the channel is decided to be a hole, the search is terminated and otherwise the process is carried on until a hole is spotted. As the spectrum holes become less frequent (sparser) enhancing the agility of the secondary users becomes more significant It is noteworthy that when the spectrum holes are {\em not} sparse, the quickest detection approach \cite{Lai:IT10_Submitted} is very effective. In such scenarios the {\em expected} density of the spectrum holes is high and it becomes very likely to encounter a spectrum hole  after sensing a few channels. Therefore, by deploying the sequential quickest detection approach the {\em expected} amount of sampling resources required for identifying multiple holes is small and thereof the overall detection process will be very agile. For a sparse distribution of spectrum holes, however, sensing the channels sequentially for identification of {\em multiple} holes will substantially lengthen the {\em expected} sensing time. Therefore, the objective of our proposed adaptive procedure is to address the issue of agility for a sparse distribution of spectrum holes and when we are interested in identifying {\em multiple} holes, as it is more challenging and realistic in practice. We also remark on the recent results on sequential binary hypothesis testing of multiple sequences \cite{Malloy:ISIT11} where a sparse number of the sequences are generated according to the null distribution and the majority of them according to an alternative distribution. It develops a sequential thresholding-based test for identifying {\em almost} all the sequences generated according to the null distribution. These tests tend to identify {\em almost} all of the sequences generated by the null distribution. 

In another approach it is assumed that the wideband channel is heavily under-utilized and it is used only rarely and sparsely. In this approach the cognitive radios exploit the sparsity structure of the wideband channel and construct compressed sensing-based machineries for estimating the power spectral density (PSD) of the wideband channel~\cite{Tan:ICASSP07, Leus:ICASSP09, Wang:GC10, Vilar:ICASSP10, Hong:GC10, Yu:ICASSP08, Tian:GC08}. These approaches are further extended to also track temporal variations of the spectral occupancy during sensing~\cite{Angelosante:ICASSP09, Angelosante:DSP09}. Exploiting the sparsity empowers the cognitive radios to sample the signal activity over the channel at a sub-Nyquist rate, which expedites the process of estimating the PSD. While sampling at a sub-Nyquist rate substantially reduces the number of samples required for reconstructing the PSD, the number of samples required yet scales linearly with the bandwidth of signals of the active users. Therefore, the amount of required samples for reconstructing the PSD of large-bandwidth signals (e.g., video signals) is substantially large. When each of the channels is narrowband, however, estimating the PSD encompasses much redundancy as the ultimate objective of the cognitive user is to make a 1-bit decision about a channel (hole vs. occupied). By taking this fact into account, our proposed sensing procedure further reduces the sampling resources below the level required by the compressed-sensing-based methods.

In another approach tailored for cognitive {\em networks} a group of cognitive radios are clustered to collaboratively perform spectrum sensing~\cite{Giannakis:SP10, Quan:SPM08, Vilar:ICASSP10, Wang:ITA09, Zeng:ICC10}. At the cost of extra communication among the cognitive radios, they can share their individual perception about the spectral occupancy and reach a more reliable decision through exchanging their individual perceptions (coordination). Such information exchange is viable in cognitive networks with a fusion center or an established infra-structure through which the autonomous users can be coordinated. In this paper we consider {\em ad-hoc} cognitive networks, in which each cognitive user operates autonomously and independently of the rest of users. In such a scenario, the task of reliable spectrum sensing must be carried out by each user independently. It is noteworthy that our approach can also be integrated with the existing collaborative sensing methods.

\subsection{Contribution}
In this paper we propose a novel adaptive spectrum sensing procedure for ad-hoc cognitive networks.  The underlying premise is that one can adaptively decide how to spend a given sampling budget in the course of the spectrum sensing process, focusing the sensing resources on more promising segments of the spectrum. This is in contrast to the conventional spectrum sensing schemes with a pre-defined sensing strategy. Our proposed adaptive procedure consists of two phases, namely the {\em exploration} and {\em detection} phases. The errand of the exploration phase is to swiftly eliminate a substantial portion of the occupied channels and at the same time retain {\em most} of the holes. Exploration is carried out through an iterative process, where in each iteration a group of less promising segments of the spectrum are identified and eliminated. The fundamental basis of the exploration phase is that it can be done accurately even if the sensed information (measurements) is very rough and noisy. During exploration, each iteration further monitors the spectrum and improves on the outcome of the previous iteration. The exploration phase is subsequently followed by a detection phase for identifying multiple spectrum holes among the reduced number of channels retained.

The analysis reveals that, if we perform multiple cycles of the proposed explorations, for achieving a certain level of reliability, the ratio of the sampling budget required by the proposed adaptive procedure to that required by the non-adaptive procedure is approximately $\frac{2}{M}$, where $M$ is the sampling budget per channel in the non-adaptive procedure. Furthermore, the adaptive procedure requires less-stringent constraints on the power of the active users to guarantee a successful hole detection. In particular, let $\epsilon_n$ be the probability that each channel is a spectrum hole. Then the requirement for identifying spectrum holes by the non-adaptive procedure is that the power of the active users must scale\footnote{For convenience we will use the following asymptotic notations: for two sequences $a_n>0$ and $b_n>0$ we say that $a_n=\omega(b_n)$ if $\lim\inf_{n\rightarrow\infty} \frac{a_n}{b_n} =\infty$ and $a_n=o(b_n)$ if $\lim\sup_{n\rightarrow\infty} \frac{a_n}{b_n} =0$. Also $a_n\doteq b_n$ indicates asymptotic equality of $a_n$ and $b_n$, i.e., $\lim_{n\rightarrow\infty}\frac{a_n}{b_n}=1$.} as $\omega\left(\sqrt[M]{\frac{1}{\epsilon_n}}\right)$, where $M$ is the sampling budget per channel. By deploying the adaptive procedure this requirement reduces to approximately\footnote{The precise characterization is given in Section~\ref{sec:adaptive}.} $\omega\left(\sqrt[M']{\frac{1}{\epsilon_n}}\right)\ $, where $M'$ is an integer substantially larger than $M$.

The exploration phase and the idea of successively searching for the spectrum holes is inspired by the notion of \emph{distilled sensing} developed in~\cite{haupt:10_submitted, haupt:08, haupt:09}. In distilled sensing it is shown that given a fixed sampling budget certain signals that are detectable/estimable using adaptive measurements cannot be recovered using non-adaptive strategies. The results show that closing the loop between the data analysis and collection processes can yield very significant gains. Despite the similarities between our exploration phase and distilled sensing in~\cite{haupt:10_submitted, haupt:08, haupt:09}, there exist also significant differences mainly in two aspects:
\begin{enumerate}
  \item We desire to identify any {\em arbitrary} number (and not necessarily all) of the spectrum holes where distilled sensing aims to identify {\em all} data points that bear a certain sparsity structure.
  \item The observation models are different. The works in \cite{haupt:10_submitted, haupt:08, haupt:09} considers a Gaussian observation model and the theoretical guarantees depend on the Gaussian assumption. On the other hand, in the spectrum sensing problem in hand we deal with data drawn from Gamma distributions. The analysis for the Gamma distribution and Gaussian distributions are significantly different. This is due to the fact that the  experimental design for distilled sensing is primarily shaped up by the statistical model of the data samples, which in turn necessitates independent analysis for each class of distributions.
\end{enumerate}
The remainder of this paper is organized as follows. In Section~\ref{sec:descriptions} the system model and the formal statement of the problem are provided. Section~\ref{sec:non_adaptive} introduces the non-adaptive hole detection scheme, which serves as a baseline for assessing the performance of the proposed adaptive procedure. The non-adaptive hole detection scheme also represents the structure of the detection phase that we deploy in the adaptive procedure. In Section~\ref{sec:adaptive} we describe the proposed adaptive hole detection procedure and its related analysis. Simulation results are provided in Section~\ref{sec:simulations} and Section~\ref{sec:conclusions} provides some concluding remarks.

\section{System Descriptions and Problem Statement}
\label{sec:descriptions}

\subsection{System Setup}
\label{sec:problem_formulation}

Consider a {\em wideband} spectrum to be shared by license-holding (primary) users and cognitive (secondary) users with interference-avoiding spectrum access. The primary users are formally allocated a number of channels in the spectrum and have the right to instantaneously access the channels whenever desired. The secondary users, on the other hand, are allowed to opportunistically seek the portions of the spectrum under-utilized by the active users (either primary users or secondary ones currently using the spectrum) and access them. Our objective is to provide an \emph{agile} and \emph{reliable} mechanism for identifying such communication opportunities for the secondary users.

We assume that the available spectrum consists of $n$ non-overlapping channels, indexed by $\{1,\ldots,n\}$. At a given instance, the active users communicate over some of these channels, which we refer to as {\em occupied channels},  and under-utilize the rest, which we call {\em spectrum holes}. We assume that the occupancy status of the spectrum remains unchanged during spectrum sensing process and and consider a probabilistic model for the occupancy of the spectrum. Let the Bernoulli random variable $Z_i$, for $i\in\{1,\ldots,n\}$, indicate the occupancy state of the $i^{th}$ channel, where $Z_i=1$ means that the $i^{th}$ channel is occupied and $Z_i=0$ states that the $i^{th}$ channel is a spectrum hole. We assume that the channels have statistically independent occupancy states and each channel is a spectrum hole with probability $\epsilon_n$. Therefore,
\begin{equation}\label{eq:ber}
    Z_i\;\stackrel{\mbox{\tiny i.i.d.}}{\sim}\; \ber(1-\epsilon_n)\ .
\end{equation}
Let us also denote the set of the indices of the spectrum holes and the occupied channels by
\begin{equation}\label{eq:H}
    \H_0\dff\big\{i\in\{1,\ldots,n\}:Z_i=0\big\}\ ,
\end{equation}
\begin{equation}    
 \H_1\dff\big\{i\in\{1,\ldots,n\}:Z_i=1\big\}\ .
\end{equation}
In order to capitalize on the availability of spectrum holes, the secondary users monitor the spectrum via channel measurements. Each measurement of channel $i$, denoted by $X_i$, is of the form
\begin{equation}
\label{eqn:single_observation}
  X_i\dff\sqrt{p_i}\; H_i\cdot S_i\cdot Z_i+W_i,\quad\mbox{for}\quad i\in\{1,\dots,n\}\ ,
\end{equation}
where $p_i$ accounts for the combined effect of the transmission power of the user active on the $i^{th}$ channel and its associated signal attenuation (for $i\in\H_1$), $H_i$ denotes the flat-fading channel from the active user to the secondary user and is distributed as zero-mean unit-variance complex Gaussian. Also, $S_i$ denotes the normalized unit-power signal of the active user active on the $i^{th}$ channel, for $i\in\H_1$. Finally $W_i$ denotes the additive white Gaussian channel noise distributed as zero-mean unit-variance complex Gaussian. The measurements are statistically independent over time and channels.


Finally, we assume that the powers of the active users sensed by the secondary user are lower-bounded by $\gamma_n$ for some arbitrary $\gamma_n>0$, i.e., $p_i\geq\gamma_n$ for all $i\in\H_1$. Acquiring $\gamma_n$ is feasible for secondary users if the geographical extent (and subsequently the maximum path-loss) and the range of the {\em transmission} power of the active users are known to them. Clearly the likelihood of channel vacancy $\epsilon_n$ and the lower bound $\gamma_n$ influence the precision of successfully detecting a hole; increasing $\gamma_n$ and $\epsilon_n$ enhances the reliability of the spectrum sensing process.

\subsection{Spectral Monitoring Goal}
\label{subsec:goal}

The main objective is to use measurements of form (\ref{eqn:single_observation}) in order to identify $T\in\mathbb{N}$ spectrum holes, i.e., $T$ element of $\H_0$. The conventional non-adaptive spectrum sensing procedures perform some pre-specified measurements of the spectrum and possibly locate $T$ spectrum holes based on the information extracted from the measurements. This strategy is non-adaptive, in the sense that the measurement process is fixed \emph{a priori} and does not change during the experiment. In contrast, we devise an adaptive procedure such that the measurement strategy is adjusted sequentially such that future measurements use information gathered from previous ones. We demonstrate that such measurement adaptation substantially improves the {\em reliability} and {\em agility} of the secondary users in detecting spectrum holes.

We assess the agility by quantifying the amount of time required for sensing the channel. Agility depends linearly on the aggregate amount of measurements (sampling budget) made over the entire wideband spectrum, assuming that the measurements take the same amount of time. Therefore, increasing the sampling budget translates into more delay (or less agility) in detection. On the other hand, increased sampling budget improves the detection reliability. Therefore, for detecting multiple spectrum holes there exists a tradeoff between agility and reliability. Incorporating the impact of channel occupancy likelihood (characterized by $\epsilon_n$) and the power of active users establishes that in both adaptive and non-adaptive procedures, there exists an inherent interplay among the agility, reliability, channel occupancy likelihood, and the power of active users.

We characterize such interplays in both non-adaptive and the proposed adaptive sensing procedures. Comparing these interplays subsequently demonstrates the agility and reliability gains afforded by the adaptive procedure. As we consider wideband channels, the analysis provided are asymptotic with respect to large number of channels, $n$.

\section{Non-Adaptive Spectrum Detection}
\label{sec:non_adaptive}

In this section we analyze an optimal non-adaptive detection scheme for locating $T$ spectrum holes. The development of this scheme and the ensuing analysis serve a two-fold purpose. On one hand this detection scheme is also deployed in the detection phase of the adaptive procedure proposed in Section~\ref{sec:adaptive}. On the other hand, it offers a baseline for assessing the gain yielded by the adaptive procedure.

\subsection{Non-Adaptive Sensing Procedure}
Constructing a non-adaptive spectrum sensing procedure involves two issues. The first one pertains to the experimental design, which is the design of the information-gathering process. In our setup the experimental design elucidates the distribution of the sampling budget among the channels across the wideband spectrum. The second issue is to design a detector based on some optimality criterion. The decision on the experimental design hinges on the expected proportion of spectrum holes, i.e., $\epsilon_n$. Prior to the sensing procedure, there is no extra side information about spectral activity in any of the channels and all the channels are equally likely to be spectrum holes. Due to the inherent symmetry and the sparse distribution of the spectrum holes ($\epsilon_n\in o(1)$), we assume that the experimental design measures all channels equally (note that if $\epsilon_n$ is large, this is not the best strategy).

Given this experimental design, it is straightforward to construct a detector that is optimal in the sense that it maximizes the {\em a posteriori} probability of successfully detecting $T$ spectrum holes, i.e., the maximum {\em a posteriori} (MAP) detector. Suppose that the measurement budget is $B\in{\mathbb N}$, meaning that we make a total of $B$ measurements of the form~\eqref{eqn:single_observation}. Therefore, the experimental design obtains $M\dff\lfloor B/n\rfloor$ measurements of the form \eqref{eqn:single_observation} per channel. Define $X_i(j)$ as the $j^{th}$ measurement of the $i^{th}$ channel. Recalling the assumption that the occupancy statuses do not change over the course of sensing, the terms $\{p_i\}$ and $\{Z_i\}$ remain constant for all $j\in\{1,\dots, M\}$, while the fading coefficient $H_i(j)$, the noise term $W_i(j)$, and the transmitted symbols $S_i(j)$ might change from one observation to another. Defining the vectors $\bX_i\dff[X_i(j)]_j$, $\bbH_i\dff[H_i(j)]_j$, $\bS_i\dff[S_i(j)]_j$, and $\bW_i\dff[W_i(j)]_j$, based on \eqref{eqn:single_observation} the observation set is given by
\begin{equation}\label{eq:D}
    \D_n\dff\big\{\bX_i:\bX_i=\sqrt{p_i}\ Z_i\ \bbH_i\circ \bS_i+\bW_i \;\;\mbox{for}\;\; i=1,\dots,n\big\} ,
\end{equation}
where ``$\circ$'' denotes the Hadamard product. From \eqref{eqn:single_observation} it follows that given the occupancy status $Z_i$, the observation sample $X_i(j)$ has a complex Gaussian distribution with mean zero and variance $(1+p_iZ_i)$, i.e.,
\begin{equation}\label{eq:X_dist}
    X_i(j)\med Z_i\stackrel{\mbox{\tiny i.i.d.}}{\sim}{\cal N}_{\mathbb C}(0,1+p_iZ_i),\quad\mbox{for}\quad j\in\{1,\dots,M\}\ .
\end{equation}
According to \eqref{eq:ber}, $Z_i$ has a Bernoulli distribution with mean $1-\epsilon_n$. Therefore the measurement vector $\bX_i$ is a sample from a mixture distribution with the probability density function
\begin{eqnarray}\label{eq:mixture}
    \nonumber P(\bX_i)&=&\epsilon_n\cdot\frac{1}{\pi^M}\exp{\left(-\|\bX_i\|^2\right)}\\
    &+& (1-\epsilon_n)\cdot\frac{1}{(\pi(1+p_i))^M} \exp{\left(-\frac{\|\bX_i\|^2}{1+p_i}\right)}\ .
\end{eqnarray}
Given the observation set $\D_n$, the maximum {\em a posteriori} (MAP) criterion for identifying $T$ spectrum holes entails finding the subset of channels with cardinality $T$ that has the highest probability of containing only holes and is formalized in the next remark.
\begin{remark}\label{rem:MAP}
The MAP rule for detecting $T$ holes is given by
\begin{eqnarray}\label{eq:MAP}
\nonumber \widehat {\cal U}_{\tiny{\mbox{MAP}}}^{\tiny{\mbox{NA}}} &\dff& \arg\max_{{\cal U}:\ |{\cal U}|=T} P({\cal U}\subseteq \H_0\med \D_n)\\
& = &\arg\min_{{\cal U}:\ |{\cal U}|=T}\sum_{i\in {\cal U}}U_i\ ,
\end{eqnarray}
where
\begin{equation}\label{eq:U}
    U_i\dff\frac{1}{(1+p_i)^M}\exp{\left(\frac{p_i}{1+p_i}\|\bX_i\|^2\right)}.
\end{equation}
\end{remark}
\begin{proof}
See Appendix \label{app:lem:MAP}.
\end{proof}
 Due to the dependence of $\{U_i\}$ on $\{p_i\}$, the posterior probability heavily depends on the power of the active users $\{p_i\}$. The uncertainties about $\{p_i\}$ translates into uncertainty about any performance measure of interest, which in turn prevents us from obtaining a {\em globally} (over all power realizations) optimal hole detector. As a remedy, by alternatively exploiting the notion of {\em robustness} in the {\em worst-case} we can offer some worst-case guarantees. In particular, we look at the realizations of $\{p_i\}$ which yield the {\em worst-case} performance for detecting $T$ holes. For this purpose, define $\Psi(\gamma_n)$ as the class of active users satisfying the power constraint $p_i\geq\gamma_n$ for all $i\in\H_1$. Then, the worst-case detection performance in the $\Psi(\gamma_n)$ class of active users is given by
\begin{equation*}
    \min_{\{p_i\}\subseteq\Psi(\gamma_n)}P(\widehat {\cal U}_{\tiny{\mbox{MAP}}}^{\tiny{\mbox{NA}}}\subseteq\H_0)\ .
\end{equation*}
Note that for $i\in\H_1$, the $i^{th}$ channel becomes less-distinguishable from a spectrum hole as $p_i$ decreases. Therefore, due to the independence among the distinct channels, the weakest detection performance, which occurs when the non-holes are least-distinguishable from the holes, corresponds to the smallest possible choices of $\{p_i\}$, i.e., $p_i=\gamma_n$ for $i\in\H_1$. For this setting, finding the indices of the $T$ smallest elements of the set $\{U_1,\dots,U_n\}$ becomes equivalent to finding the indices of the $T$ smallest elements of the set $\{\|\bX_1\|^2,\dots,\|\bX_n\|^2\}$. Therefore, in order to locate $T$ spectrum holes the MAP detector requires only the sufficient statistics $Y_i\dff\|\bX_i\|^2$ for $i=1,\dots, n$. Corresponding to the sequence of random variables $\{Y_1,\dots,Y_n\}$ we define $\{Y_{(1)},\dots,Y_{(n)}\}$ as the sequence of order statistics in an increasing order, e.g., $Y_{(m)}$ represents the $m^{th}$ smallest element of $\{Y_1,\dots,Y_n\}$. Based on this definition, the {\em robust} hole detector (with the non-adaptive procedure) can be formalized as follows.
\begin{remark}[Robust Spectrum Detector]
\label{remark:MAP} In the $\Psi(\gamma_n)$ class of active users when  $\epsilon_n=o(1)$, the robust non-adaptive spectrum detector for identifying $T$ holes is given by
\begin{equation}\label{eq:RMAP}
\widehat {\cal U}_{\tiny{\mbox{\rm rob}}}^{\tiny{\mbox{\rm NA}}}=\{i\in\{1,\dots,n\}:Y_i\leq Y_{(T)}\}\ .
\end{equation}
\end{remark}
\subsection{Asymptotic Performance}
Recalling the distribution of $X_i(j)$ given in \eqref{eq:X_dist}, the sufficient statistics $Y_i$ is distributed as
\begin{equation*}
   Y_i\med Z_i\sim \dGamma(M,1+p_i Z_i)\quad\mbox{for}\quad i=1,\dots,n\ ,
\end{equation*}
where $\dGamma(a,b)$ denotes a Gamma distribution\footnote{The random variable $G$ with distribution $\dGamma(a,b)$, where $a,b>0$, has a density with respect to the Lebesgue measure of the form $f_G(t)=t^{a-1} \frac{\exp(-t/b)}{b^a \Gamma(a)}\ \mathbf{1}\{t\geq 0\}$.} with parameters $a$ and $b$. For assessing the performance of the {\em robust} detection scheme, for all $i\in\H_1$ we set $p_i=\gamma_n$. Clearly the robust detector will make a detection error if  $\widehat {\cal U}_{\tiny{\mbox{rob}}}^{\tiny{\mbox{NA}}}\bigcap\H_1\neq\emptyset$. Let us define $u_t$ and $v_t$ as the indices of the $t^{th}$ smallest elements of the sets $\{Y_i\}$ and $\{Y_i:i\in\H_0\}$, respectively. From Remark~\ref{remark:MAP} the worst-case detection error probability is given by
\begin{eqnarray*}
P_{{\mbox{\tiny NA}}}(n)&\dff& \min_{\{p_i\}\subseteq\Psi(\gamma_n)} P\left(\widehat {\cal U}_{\tiny{\mbox{rob}}}^{\tiny{\mbox{NA}}}\cap\H_1\neq\emptyset\right)\\
&=& 1-\max_{\{p_i\}\subseteq\Psi(\gamma_n)} P\left(\widehat {\cal U}_{\tiny{\mbox{rob}}}^{\tiny{\mbox{NA}}}\cap\H_1 =\emptyset\right)\\ &=&  1- \max_{\{p_i\}\subseteq\Psi(\gamma_n)} P(\{u_1,\dots,u_T\}\subseteq\H_0)\ ,
\end{eqnarray*}
Note that having all the $T$ smallest measured energies belonging to the holes is equivalent to not having the energies of {\em all} non-holes greater than the $T$ smallest measured energies, i.e.,
\begin{align*}
    \{\{u_1,\dots,u_T\}&\in \H_0\}\\
    &\equiv\{Y_{v_t}\leq Y_i\ ,\;\forall t\in\{1,\dots,T\}\;\mbox{and}\;\forall i\in\H_1\}\\
    &\equiv \left\{\max_{t\in\{1,\dots,T\}}Y_{v_t}\leq\min_{i\in\H_1}Y_i\right\}\\
    &\equiv \left\{Y_{v_T}\leq\min_{i\in\H_1}Y_i\right\}\ .
\end{align*}
Therefore,
\begin{equation}\label{eq:P_NA2}
    P_{{\mbox{\tiny NA}}}(n)= 1-\max_{\{p_i\}\subseteq\Psi(\gamma_n)} P(Y_{v_T}\leq \min_{i\in\H_1}Y_i)\ ,
\end{equation}
Assessing $P_{{\mbox{\tiny NA}}}(n)$ as defined above relies on the properties of the order statistics of a set of random variables. The following lemma is instrumental for characterizing the distributions of the order statistics and evaluating $P_{{\mbox{\tiny NA}}}(n)$. This is a generalization of a well-studied problem in the context of extreme value theory that considers {\em the first} order statistic \cite{David:book} for Gamma distributions. In this lemma, we give the corresponding results when distribution evolution is allowed, i.e., the {\em number} and {\em distribution} of the involved random variables changes {\em simultaneously},  as well as the analysis for higher order statistics, for which the existing results are not applicable.
\begin{lemma}
\label{lmm:extreme}
Let $\{Y_i\}_{i=1}^{m}$ be a sequence of i.i.d. random variables distributed as $\dGamma(M,\alpha_m)$ and denote its corresponding sequence of order statistics by $\{Y_{(i)}\}_{i=1}^m$. Let $b_m\dff\alpha_m\left[\frac{\Gamma(M+1)}{m}\right]^{\frac{1}{M}}$ and for some $T\in\mathbb{N}$ define the sequence of random variables $W_{(i)}^m\dff\frac{Y_{(i)}}{b_{m}}$ for $i=1,\dots,T$. Then as $m\rightarrow\infty$, $W_{(i)}^m$ converges in distribution to a random variable $W_{(i)}$ with cumulative density function (CDF)
\begin{equation*}
    Q_{(i)}(w;m)\dff P(W_{(i)}<w)\doteq 1-\exp(-w^M)\sum_{k=0}^{i-1}\frac{w^{kM}}{k!}\ .
\end{equation*}
\end{lemma}
\begin{proof}
See Appendix \ref{app:lmm:extreme}.
\end{proof}
\begin{figure*}[!b]
\hrulefill
\begin{eqnarray*}
\nonumber \lefteqn{\min_{\{p_i\}\subseteq\Psi(\gamma_n)} P\left(\left.\widehat {\cal U}_{\tiny{\mbox{rob}}}^{\tiny{\mbox{NA}}}\cap\H_1\neq\emptyset \right| \{Z_i\}\right)}\\
  &=&  1-P\left(b_n^0 W^n_{(T);0} < b_n^{1} W^n_{(1);1}\right) \\
  &=& 1- P\left( W^n_{(T);0} < \frac{b_n^{1}}{b_n^0}\cdot W^n_{(1);1}\right)\\
 & = & 1- \int_0^\infty q_{(1);1}(x;n_{1}) \int_0^{\frac{b_n^{1}}{b_n^0} x} q_{(T);0} (y;n_0)\  dy \ dx\\
 & = & 1- \int_0^\infty q_{(1);1}(x;n_{1})\ Q_{(T);0}\left(\frac{b_n^{1}}{b_n^0}\cdot x;n_0\right)\ dx\\
 & \doteq & 1- \int_0^\infty \underset{=Mx^{M-1}e^{-x^M}}{\underbrace{q_{(1);1}(x;n_{1})}} \left(1-\exp\left(-x^M \left(b_n^{1}/ b_n^0\right)^M\right)\sum_{k=0}^{T-1}\ \frac{x^{kM}\left(b_n^{1}/ b_n^0\right)^{kM}}{k!}\right)\ dx\\
&= & 1- \left\{1-\sum_{k=0}^{T-1}\frac{M\left(b_n^{1}/ b_n^0\right)^{kM}}{k!}\int_0^\infty \exp\left(-x^M\left(1+ \left(b_n^{1}/ b_n^0\right)^M\right)\right)\ x^{kM+M-1}\ dx\right\}\ ,
\end{eqnarray*}
By setting
\begin{equation*}
    s\dff x^M\left(1+ \left(b_n^{1}/ b_n^0\right)^M\right)\ ,
\end{equation*}
we further find
\setcounter{MYtempeqncnt}{\value{equation}}
\setcounter{equation}{15}
\begin{eqnarray}
\label{eq:below1} 
\nonumber \lefteqn{\min_{\{p_i\}\subseteq\Psi(\gamma_n)} P\left(\left.\widehat {\cal U}_{\tiny{\mbox{rob}}}^{\tiny{\mbox{NA}}}\cap\H_1\neq\emptyset \right| \{Z_i\}\right)}\\
\nonumber &\doteq & 1- \left\{1-\sum_{k=0}^{T-1}\frac{M\left(b_n^{1}/ b_n^0\right)^{kM}}{k!}\cdot\frac{1}{M\left(1+\left(b_n^{1}/ b_n^0\right)^{M}\right)^{k+1}}\underset{=\ \Gamma(k+1)\ =\ k!}{\underbrace{\int_0^\infty \exp(-s)\  s^k\ ds}}\right\}\\
  & = & 1- \left\{1-\sum_{k=0}^{T-1}\frac{\left((1+\gamma_n)^M\cdot\frac{n_0}{n_1}\right)^{k}} {\left(1+(1+\gamma_n)^M\cdot\frac{n_0}{n_1}\right)^{k+1}}\right\} \ .
\end{eqnarray}
\setcounter{equation}{\value{MYtempeqncnt}}
\end{figure*}

For the setting of Section \ref{sec:problem_formulation}, the following theorem characterizes the asymptotic performance of the robust hole detector in the $\Psi(\gamma_n)$ class of active users. It also establishes the tradeoffs among the spectrum vacancy likelihood $\epsilon_n$, per channel sampling budget $M$, and the minimum power of active users.

\begin{theorem}[Non-Adaptive Tradeoff]
\label{thm:NA_bound}
For the $\Psi(\gamma_n)$ class of active users, when $\epsilon_n=o(1)$ and $n\epsilon_n=\omega(1)$, the error probability of the robust detector for identifying $T$ spectrum holes is given by
\begin{eqnarray}
\label{eq:P_NA}
\nonumber P_{\mbox{\tiny NA}}(n) &=& \min_{\{p_i\}\subseteq\Psi(\gamma_n)} P\left(\widehat {\cal U}_{\tiny{\mbox{rob}}}^{\tiny{\mbox{NA}}}\cap\H_1\neq\emptyset\right)\\
& \doteq & 1- \left(1+[(1+\gamma_n)^M\cdot\epsilon_n]^{-1}\right)^{-T}\ .
\end{eqnarray}
\end{theorem}
\begin{proof}
For assessing the performance of the {\em robust} detection scheme, we set $p_i=\gamma_n$ for all $i\in\H_1$, as this corresponds to the worst case scenario. The number of holes denoted by $|\H_0|$, is a random quantity. We proceed by conditioning on $\{Z_i\}_{i=1}^n$ in what follows, but for notational convenience we do not explicitly represent this dependence. Define $n_0\dff |\H_0|$ and $n_1\dff|\H_1|=n-n_0$ (conditionally on $\{Z_i\}$ these are constants).

We are now ready to apply Lemma~\ref{lmm:extreme} on the two sequences of i.i.d. random variables $\{Y_i: i\in\H_0\}$ distributed as $\dGamma(M,1)$ and $\{Y_i : i\in\H_1\}$ distributed as $\dGamma(M,1+\gamma_n)$. Corresponding to these sequences define
\begin{equation}\label{eq:b0}
    b_n^0\dff\left[\frac{\Gamma(M+1)}{n_0}\right]^{1/M}\ ,
\end{equation}
and
\begin{equation}\label{eq:b1}
 b_n^{1}\dff(1+\gamma_n)\left[\frac{\Gamma(M+1)}{n_{1}}\right]^{1/M} \ .
\end{equation}
Also define
\begin{equation}\label{eq:W}
    W^n_{i;0}\dff\frac{1}{b_n^0}\ Y_i, \quad i\in \H_0 ,\quad\mbox{and}\quad W^n_{i;1}\dff\frac{1}{b_n^1} Y_i\ ,\quad i\in\H_1 \ .
\end{equation}
For convenience denote the probability density functions (PDF) of $W^n_{i;0}$ by $q_{i;0}(w;n_1)\dff\frac{d}{dw} Q_{i;0}(w;n_1)$ and PDF of $W^n_{i;1}$ by $q_{i;1}(w;n_1)\dff\frac{d}{dw} Q_{i;1}(w;n_1)$ (we drop the explicit dependence on $n$ to avoid notational clutter). Note that $\{W^n_{i;0}\}$ and $\{W_{i;1}^n\}$ are statistically independent. In what follows we are going to be interested in the order statistics of these sequences, in particular $W^n_{(i);0}$, $i\in\{1,\ldots,T\}$, the smallest $T$ elements of $\{W^n_{i;0}\}_{i\in\H_0}$ sorted in increasing magnitude, and $W^n_{(1);1}=\min_{i\in\H_1} W^n_{(i);1}$. Taking this into account and noting that, by the law of large numbers, the assumptions $\epsilon_n=o(1)$ and $n\epsilon_n=\omega(1)$ imply that $n_0,n_1\rightarrow\infty$ as $n\rightarrow\infty$, we have according to Lemma~\ref{lmm:extreme}
\begin{align*}
    Q_{(1);1}(w;n_1)&\doteq1-\exp(-w^M)
 \end{align*}
 \begin{align*}
Q_{(i);0}(w;n_0)&= P(W_{(i);0}^m<w)\\
&\doteq 1-\exp(-w^M)\sum_{k=0}^{i-1}\frac{w^{kM}}{k!}\ .
\end{align*}
As a result, by invoking \eqref{eq:P_NA2} when $n_0$ and $n_1$ approach infinity
\begin{equation*}
\min_{\{p_i\}\subseteq\Psi(\gamma_n)} P\left(\left.\widehat {\cal U}_{\tiny{\mbox{rob}}}^{\tiny{\mbox{NA}}}\cap\H_1\neq\emptyset \right| \{Z_i\}\right)
\end{equation*}
is characterized in \eqref{eq:below1}. 
We now remove the conditioning on $\{Z_i\}$, by computing the expectation of the expression above. Recall that $n_0/n_1$ is a random quantity, and by assumption and the law of large numbers $n_0/n_1 \rightarrow\epsilon_n$. Using this fact and continuous mapping principles we conclude that
\begin{align*}
\lefteqn{\min_{\{p_i\}\subseteq\Psi(\gamma_n)} P\left(\left.\widehat {\cal U}_{\tiny{\mbox{rob}}}^{\tiny{\mbox{NA}}}\cap\H_1\neq\emptyset \right| \{Z_i\}\right)}\\
& \doteq  1- \left\{1-\sum_{k=0}^{T-1}\frac{\left((1+\gamma_n)^M\cdot\epsilon_n\right)^{k}} {\left(1+(1+\gamma_n)^M\cdot\epsilon_n\right)^{k+1}}\right\} \\
& =  1- \left\{1-\frac{1}{1+(1+\gamma_n)^M\cdot\epsilon_n}\cdot \frac{1-\left(\frac{(1+\gamma_n)^M\cdot\epsilon_n} {1+(1+\gamma_n)^M\cdot\epsilon_n}\right)^T} {1-\left(\frac{(1+\gamma_n)^M\cdot\epsilon_n} {1+(1+\gamma_n)^M\cdot\epsilon_n}\right)}\right\}\\
& =  1- \left(1+[(1+\gamma_n)^M\cdot\epsilon_n]^{-1}\right)^{-T}\ ,
\end{align*}
as desired.
\end{proof}

As expected, there exists a tension between reliability and agility. On one hand, increasing the sampling budget per channel $M$ favors reliability, as according to \eqref{eq:P_NA} it improves the probability of successfully detecting a hole. Increasing the sampling budget, on the other hand, imposes more delay in spectrum sensing as the time required for completing the acquisition of data scales linearly with the sampling budget. Finally, \eqref{eq:P_NA} also states that an increase in the minimum power of the active users $\gamma_n$, improves the reliability. Theorem~\ref{thm:NA_bound} characterizes the tradeoff among reliability, hole detection agility, and the power of active users. By using the result of Theorem~\ref{thm:NA_bound}, we offer a necessary and sufficient condition on the scaling of the power of the active users to guarantee asymptotically error-free multiple hole detection in the non-adaptive sensing setting.
\begin{corollary}[Non-Adaptive Power Scaling]
\label{cor:power_NA} For the $\Psi(\gamma_n)$ class of active users, when $\epsilon_n=o(1)$ and $n\epsilon_n=\omega(1)$, the necessary and sufficient condition for $P_{\mbox{\tiny NA}}(n)\rightarrow 0$ as $n\rightarrow \infty$ is that $\gamma_n$ scales with increasing $n$ as
\setcounter{equation}{16}
\begin{equation}
\gamma_n= \omega\left(\sqrt[M]{\frac{1}{\epsilon_n}}\right)\ .
\end{equation}
\end{corollary}
In other words, if the power of the ``faintest'' active user grows faster than $\sqrt[M]{\frac{1}{\epsilon_n}}\;$, then a secondary user can reliably identify $T$ holes by employing this non-adaptive procedure. The proof follows in a straightforward way from the characterization of $P_{\mbox{\tiny NA}}(n)$ given in \eqref{eq:P_NA} by noting that $P_{\mbox{\tiny NA}}(n)\rightarrow 0$ is equivalent to $(1+\gamma_n)^M\cdot\epsilon_n\rightarrow\infty$.

\section{Adaptive Spectrum Detection}
\label{sec:adaptive}

\subsection{Adaptive Sensing Procedure}

Our proposed adaptive spectrum sensing procedure has two phases, namely the {\em exploration} phase and the {\em detection} phase. The exploration phase, being an iterative procedure, is intended to purify the set of the channels to be sensed carefully for detecting  spectrum holes. This phase is accomplished by successively identifying and eliminating a group of channels deemed to be occupied. The detection phase is performed after the exploration phase in order to identify $T$ holes among the subset of candidate channel retained after exploration. The detection scheme deployed is identical to the robust spectrum detection scheme of Section~\ref{sec:non_adaptive}.

The exploration phase proceeds in an iterative way. In each iteration it further monitors the channels retained by the previous iteration and eliminates those deemed to be spectrum holes least-likely. The core idea is that it is relatively easy to identify occupied channels with low-quality measurements since there are few holes available (recall that $\epsilon_n$ is small). Each iteration carries on by thresholding the observed energy on each channel retained by the previous iteration. The threshold level depends only on $\gamma_n$, and is designed such that at each iteration roughly half of the existing occupied channels are eliminated, while almost all the of spectrum holes are preserved. The output of each exploration phase will have a more condensed proportion of spectrum holes to occupied channels. Subsequently, the detector developed for the non-adaptive procedure is applied on this refined set of channels in order to identify $T$ spectrum holes. This entire procedure bears similarities with Distilled Sensing~\cite{haupt:10_submitted, haupt:08, haupt:09}, however, the analysis is substantially different. This is due to the different sensing objective (identifying any arbitrary number of holes as opposed to \cite{haupt:10_submitted, haupt:08, haupt:09} that aims to identify almost all) as well as the underlying statistical model.

We show that the gains yielded by this adaptive procedure can be interpreted in two ways. First we demonstrate that when targeting the same level of reliability in hole detection, the adaptive procedure requires substantially less sampling budget, or equivalently it is substantially more agile. Secondly, we show that under the same sampling budget, and targeting identical hole detection reliability, the adaptive procedure imposes less-stringent conditions on how fast the power of the active users $\gamma_n$ must scale with increasing $n$. This essentially indicates that for some choices of $\gamma_n$ the adaptive procedure can guarantee successful hole detection while the best non-adaptive procedure fails to do so.

Let us define $K$ as the number of exploration cycles (iterations) in the exploration phase. Also denote the sampling budget per channel in the $k^{th}$ exploration cycle by $M_k$. The exploration phase is initialized by including {\em all} channels for sensing and resumes as follows. In the first iteration all channels are allocated the identical sampling budget of $M_1$. The energy levels of all channels are compared against $\lambda_1(1+\gamma_n)$, where $\lambda_1$ is the median of the distribution $\dGamma(M_1,1)$. The channels for which the measured energy exceed this threshold are discarded and the rest are carried over to the second iteration for further sensing. The same procedure is repeated throughout all $K$ cycles. More specifically, in the $k^{th}$ cycle all the channels retained by the $(k-1)^{th}$ iteration are allocated the identical sampling budget of $M_k$. The energy levels of these channels are compared with $\lambda_k(1+\gamma_n)$, where $\lambda_k$ is the median of the distribution $\dGamma(M_k,1)$ and the exploration is performed via thresholding as in the first iteration. Finally, after the exploration phase, each of the remaining channels is allocated the sampling budget of $M_{K+1}$ and the robust spectrum detection scheme provided in Remark \ref{remark:MAP} is applied in order to detect $T$ holes.

We set $\G_0\dff\{1,\dots,n\}$ and for $k=1,\dots,K$, we define $\G_k$ as the set of the indices of the channels that are retained by the $k^{th}$ exploration cycle. Clearly we have $\G_{K}\subseteq \dots, \subseteq\G_1\subseteq \G_0$ and $\G_{K}$ contains the set of the indices of the candidate channels among which $T$ holes will be detected. The set of measurements defined for the non-adaptive scheme in (\ref{eq:D}) is extended for the proposed adaptive procedure as follows. We define the set of measurements as
\begin{equation*}
    \D_n^k\dff\big\{\bX^k_i:\bX^k_i=\sqrt{p_i}\;Z_i\ \bbH^k_i\circ \bS^k_i\cdot +\bW^k_i \quad\mbox{for}\quad i\in\G_{k-1}\big\}\ ,
\end{equation*}
for $ k=1,\dots,K+1$.  The measurement sets $\D_n^1,\dots,\D_n^{K}$ are processed in the exploration phase and the measurement set $\D_n^{K+1}$ is used in the detection phase. Based on the model of channels dynamics described in Section~\ref{sec:problem_formulation}, only the spectral occupancy $\{Z_i\}$ and the power of active users $\{p_i\}$ remain unchanged during all observations, and fading, channel noise and transmitted signal change to independent states after each  observation. Therefore, given the occupancy status $Z_i$, the observation sample $X^k_i(j)$ for $k=1,\dots,K+1$ is distributed as
\setcounter{equation}{17}
\begin{equation}\label{eq:X_dist_adaptive}
    X^k_i(j)\med Z_i\stackrel{\mbox{\tiny i.i.d.}}{\sim}{\cal N}_{\mathbb C}(0,1+p_iZ_i)\ ,
\end{equation}
for $i\in\G_{k-1}$ and $j=1,\dots,M$. We also define
\begin{equation}\label{eq:Y_adaptive}
    Y_i^k\dff\|\bX^k_i\|^2\quad\mbox{for}\quad i\in \G_{k-1}\quad\mbox{and}\quad k=1,\dots,K+1\ .
\end{equation}
Equations \eqref{eq:X_dist_adaptive} and \eqref{eq:Y_adaptive} provide that for a given $Z_i$, $Y^k_i$ is distributed as
\begin{equation}\label{eq:Y_dist_adaptive}
   Y^k_i\med Z_i\sim \dGamma(M_k,1+p_i Z_i)\ ,
\end{equation}
for $i\in\G_{k-1}$ and $k=1,\dots,K+1$.
For each $k=1,\dots, K+1$, corresponding to the sequence $\{Y^k_i\}_{i\in\G_{k-1}}$ we define the sequence of order statistics $\{Y^k_{(i)}\}_{i\in\G_{k-1}}$ in an increasing order such that $Y^k_{(i)}$ represents the $i^{th}$ smallest element of this sequence. The adaptive sensing procedure is formally described in Table 1.

\begin{table*}
\label{fig:algorithm}
\rule{\linewidth}{0.3mm}\vspace{.1in}
\begin{minipage}[h]{6 in}
{\normalsize{{\rm
\begin{tabular}{ll}
   &\hspace{2.5 in}{\it Exploration phase} \vspace{-.1in} \\
   &\rule{\linewidth}{0.1mm}\vspace{.00in}\\
   \;1:& \textbf{Input} $K\in\mathbb{N}$ and $\{M_1,\dots, M_{K+1}\}$ where $M_k\in\{1,2,3,\dots\}$.\\
   \;2:& \textbf{Initialize} the index set $\G_0\leftarrow\{1,\dots,n\}$.\\
   \;3:& \textbf{\bf for} $k=1,\dots,K$ \textbf{\bf do}\\
   \;4:& \quad Set
   $Y_i^k=\left\{\begin{array}{ll}
      \|\bX^k_i\|^2  & \quad\mbox{for}\quad i\in \G_{k-1} \\
      +\infty & \quad\mbox{for}\quad i\notin \G_{k-1}
    \end{array}\right.
   $
   .\\
   \;5:& \quad Obtain $\G_{k}\leftarrow \left\{i\in \G_{k-1}\med
   Y_i^k<\lambda_k(1+\gamma_n)\right\}$ where $\lambda_k $ is the median of ${\rm Gamma}(M_k,1)$.\\
   \;6:& \textbf{\bf end for}\\
   & \rule{\linewidth}{0.1mm}\vspace{.00in}\\
   &\hspace{2.5 in}{\it Detection phase} \vspace{-.1in} \\
   &\rule{\linewidth}{0.1mm}\vspace{.00in}\\
   \;7:&  Set
   $Y_i^{K+1}=\left\{\begin{array}{ll}
      \|\bX^{K+1}_i\|^2  & \quad\mbox{for}\quad i\in \G_{K} \\
      +\infty & \quad\mbox{for}\quad i\notin \G_{K}
    \end{array}\right.
   $
   .\\
   \;8:& Identify the indices of a spectrum holes $\widehat {\cal U}_{\tiny{\mbox{rob}}}^{\tiny{\mbox{A}}}\dff\{i\in\G_{K}: Y^{K+1}_i\leq Y^{K+1}_{(T)}\}$.\\
   \;9:& \textbf{Output} $\widehat {\cal U}_{\tiny{\mbox{rob}}}^{\tiny{\mbox{A}}}$.
\end{tabular}}}
}
\end{minipage}\\
\rule{\linewidth}{0.3mm}
\caption{adaptive robust spectrum detection algorithm}
\end{table*}

\subsection{Asymptotic Performance}
We start by assessing the performance for any given value of the exploration cycles $K$ and relegate the discussions on the optimal design of $K$ to the next section. The analysis of the adaptive sensing procedure follows the approach of \cite{haupt:10_submitted, haupt:08}, albeit with the non-trivial modifications to deal with the different objective and the different observation model. The following lemmas shed light on how the adaptive procedure accomplishes the exploration cycles. Lemma \ref{lmm:hole} characterizes the proportion of the spectrum holes that are retained in each exploration cycle.
\begin{lemma}\label{lmm:hole}
Let $m_0=|\H_0|$ and for $k=1,\dots,K$ define $m_k$ as the number of holes retained by the $k^{th}$ exploration cycle. Conditionally on $m_{k-1}$ for $ k=1,\dots,K$ and for sufficiently large $n$ the event
\begin{equation}\label{eq:m_k}
    \left(\frac{\gamma_n}{1+\gamma_n}\right)m_{k-1}\leq  m_{k} \leq m_{k-1}\ ,
\end{equation}
holds with probability at least $1-\exp\left(-\frac{m_{k-1}}{n^\alpha}\right)$ for any $\alpha>0$.
\end{lemma}
\begin{proof}
See Appendix \ref{app:lmm:hole}.
\end{proof}
The next lemma shows that during each exploration cycle almost half of the non-holes are eliminated.
\begin{lemma}\label{lmm:occupied}
Let $\ell_0=|\H_1|$ and for $k=1,\dots,K$ define $\ell_k$ as the number of the occupied channels retained by the $k^{th}$ exploration cycle. Conditionally on $\ell_{k-1}$ for $ k=1,\dots,K$ and for sufficiently large $n$ the event
\begin{equation}\label{eq:l_k}
    \left(c_k-\frac{1}{\log n}\right)\ell_{k-1} \leq  \ell_{k}\leq  \left(c_k+\frac{1}{\log n}\right)\ell_{k-1} \ ,
\end{equation}
holds with probability at least $1-2\exp\left(-\frac{2\ell_{k-1}}{(\log n)^2}\right)$, where $c_k\leq\frac 1 2$ is a constant. Furthermore, $c_k=\frac{1}{2}$ if and only if $p_i=\gamma_n$ for $i\in\H_1$.
\end{lemma}
\begin{proof}
See Appendix \ref{app:lmm:occupied}.
\end{proof}
A careful use of the above lemmas establishes the performance of the adaptive robust spectrum detection in the following theorem.
\begin{theorem}[Adaptive Tradeoff]
\label{thm:adaptive}
For the $\Psi(\gamma_n)$ class of active users, when $\epsilon_n=o(1)$ and $n\epsilon_n=\omega(1)$, the error probability of the adaptive robust spectrum detection procedure for identifying $T$ spectrum holes is given by
\begin{eqnarray}
\label{eq:P_adaptive}
\nonumber P_{\mbox{\tiny A}}(n) & \dff & \min_{\{p_i\}\subseteq\Psi(\gamma_n)} P\left(\widehat {\cal U}_{\tiny{\mbox{rob}}}^{\tiny{\mbox{A}}}\cap\H_1\neq\emptyset\right)\\
& \doteq & 1- \left(1+[(1+\gamma_n)^{M_{K+1}}\cdot2^K\epsilon_n]^{-1}\right)^{-T}\ .
\end{eqnarray}
\end{theorem}
\begin{proof}
See Appendix \ref{app:thm:adaptive}.
\end{proof}
The tradeoff above suggests that the \emph{asymptotic} performance of the adaptive procedure does not impose any constraint on the choice of $M_k$ for $k=1,\dots,K$, other than the trivial constraint $M_k\geq 1$. Hence, for achieving the best \emph{asymptotic} performance the best sequential experimental design requires maximizing $M_{K+1}$, i.e., the sampling budget per channel in the detection phase. More specifically, given a fixed sampling budget, i.e., a ``cap'' on the maximum number of measurements available, the best asymptotic strategy is to allocate as much as possible sampling budget for the detection phase and to allocate as low as {\em one} sample per channel in each exploration cycle. This implies that as low as one sample per channel in each exploration cycle suffices to ensure that for sufficiently large $n$, the exploration phase retains almost all the holes and discards almost half of the occupied channels. Therefore, to achieve the best asymptotic behavior we should set $M_k=1$ for $k=1,\dots,K$.

In order to quantify the gains yielded by the adaptive algorithm, we compare the results for the non-adaptive and adaptive schemes provided in Theorems~\ref{thm:NA_bound} and \ref{thm:adaptive}. In particular, with the aim of attaining identical asymptotic reliability levels in the adaptive and the non-adaptive procedures, i.e., $P_{\mbox{\tiny A}}(n)\doteq P_{\mbox{\tiny NA}}(n)$, we characterize the \emph{agility gain}, which we define as the ratio of the sampling budgets required by the non-adaptive procedure to that required by the adaptive scheme, i.e.,
\begin{align*}
    &\mbox{agility gain}\\
    &\dff\frac{\mbox{ sampling budget of the non-adaptive procedure}}{\mbox{ sampling budget of the adaptive  procedure}}\ ,
\end{align*}
\begin{theorem}[Agility]
\label{th:agility}
For the $\Psi(\gamma_n)$ class of active users, when $\epsilon_n=o(1)$ and $n\epsilon_n=\omega(1)$, the agility gain of the adaptive robust spectrum detection algorithm with $Mn$ sampling budget is asymptotically lower bounded by $\left(\frac{1}{2^{K}}+\frac{2}{M}\right)^{-1}$, where $K$ is the number of exploration cycles.
\end{theorem}
\begin{proof}
See Appendix \ref{app:th:agility}.
\end{proof}
It is noteworthy that while the number of exploration cycles $K$ can be made arbitrarily large (but fixed as a function of $n$), increasing it beyond some point will affect the agility very insignificantly. More specifically, for large $K$, the agility gain lower bound will be dominated by the term $\frac{M}{2}$. This underlines the fundamental limit of the agility gain yielded by the adaptive procedure.

An analogue of Corollary~\ref{cor:power_NA} can be derived for the adaptive procedure, providing a necessary and sufficient condition on the scaling of $\gamma_n$ for guaranteeing a reliable hole detection with the proposed algorithm. For comparison purposes we assume that both adaptive and non-adaptive procedures are granted the same sampling budget.
\begin{corollary}[Adaptive Power Scaling]
\label{cor:power_adaptive}
For the $\Psi(\gamma_n)$ class of active users, when $\epsilon_n=o(1)$ and $n\epsilon_n=\omega(1)$, given that the sampling budget is $Mn$, a necessary and sufficient condition for $P_{\mbox{\tiny A}}(n)\rightarrow 0$ as $n\rightarrow \infty$ is that
\begin{equation}\label{eq:gamma_adaptive}
\gamma_n= \omega\left(\sqrt[M']{\frac{1}{2^K\epsilon_n}}\right)\ ,
\end{equation}
where $M'\geq 2^{K}(M-2)+2$.
\end{corollary}
\begin{proof}
See Appendix \ref{app:cor:power_adaptive}.
\end{proof}
Comparing the result above with that of Corollary \ref{cor:power_NA} shows that an adaptive scheme can cope with much weaker active users than a non-adaptive scheme. More specifically, by noting that $M'$ is substantially larger than $M$, the power scaling requirement in the adaptive scenario, which is smaller than $\omega\left(\sqrt[M']{\frac{1}{\epsilon_n}}\right)$ becomes substantially smaller than its counterpart in the non-adaptive scenario $\omega\left(\sqrt[M]{\frac{1}{\epsilon_n}}\right)$.  As a result, there are scenarios where non-adaptive schemes fail to successfully identify $T$ holes, while the adaptive scheme succeeds.

\subsection{Optimal Exploration Cycles ($K$)}
\label{sec:K}

The role of the exploration cycles is to feed the detector with a group of channels with a more condensed proportion of spectrum holes to occupied channels. According to the characterization of $P_{\mbox{\tiny A}}(n)$ given in \eqref{eq:P_adaptive}, changing $K$ has two effects on the asymptotic behavior of $P_{\mbox{\tiny A}}(n)$. Specifically, increasing $K$ results in a decrease in $P_{\mbox{\tiny A}}(n)$ through its direct impact via the term $2^K$. Increasing the exploration cycles $K$, on the other hand, increases the sampling resources during exploration and leaves less resources for the detection phase, i.e., $M_{K+1}$ decreases. Following the same arguments as in the proof of Theorem~\ref{th:agility} it can be readily verified that the superposition of the two effects implies that $P_{\mbox{\tiny A}}(n)$ will be monotonically decreasing with increasing $K$. The results of Theorem~\ref{th:agility}~and Corollary~\ref{cor:power_adaptive} also substantiate that increasing $K$ enforces a less stringent constraint on power scaling and enhances the agility.

An important feature in the analysis of the proposed algorithm was that the exploration step asymptotically retains all the existing holes, while discarding a large number of occupied channels. This ensures the task of the detection phase is much more effective. As all the results presented so far hint that a large value of $K$ yields higher gains it is natural to ask what is the ``optimal'' value of $K$ so that the exploration phase asymptotically retains all the holes. This imposes an upper bound on possible values of $K$, and gives rise to the following theorem.

\begin{theorem}\label{th:K}
The optimal growth of the number of cycles $K$ so that the exploration phase asymptotically retains all the holes for consideration in the detection phase is given by
\begin{equation}\label{eq:K}
    K^*\doteq\log\log \frac{1}{\epsilon_n}\ .
\end{equation}
Furthermore, the outcome of Theorem~2 remains valid with any choice of $K$ not exceeding $K^*$, which is no longer fixed as a function of $n$.
\end{theorem}
\begin{proof}
See Appendix \ref{app:th:K}.
\end{proof}
\section{Simulation Results}\label{sec:simulations}
\begin{figure}[t]
\centering
\includegraphics[width=3.5in]{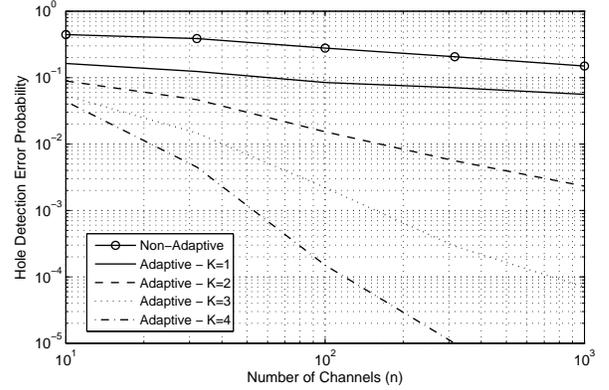}\\
\caption{Comparison of the error probabilities $P_{\tiny{\mbox{A}}}(n)$ and $P_{\tiny{\mbox{NA}}}(n)$ for detecting $T=2$ over a wide range of $n$.}\label{fig:reliability}
\end{figure}
In this section we provide some simulation results to empirically demonstrate the effectiveness of the proposed adaptive procedure. In Figure~\ref{fig:reliability} we aim at identifying $T=2$ spectrum holes and compare the reliability of the non-adaptive scheme with that of the proposed adaptive procedure with $K=1,\dots,4$ exploration cycles. Comparisons are provided over the range of $n=10-1000$ channels. The sampling budget in both schemes is set as $5n$ which means that in the non-adaptive scheme each channel is measured $M=5$ times. In the adaptive scheme in each cycle of the exploration phase each channel is measured once. The sampling resources not used in the exploration phase are equally divided among the remaining channels retained by the exploration phase. We set the channel occupancy probability $\epsilon_n=n^{-2/3}$, which clearly satisfies the conditions $\epsilon_n=o(1)$ and $n\epsilon_n=\omega(1)$, and finally assume that the power of active users are $p_i=\gamma_n=n^{1/5}$, $\forall i\in\H_1$. As predicted by the analysis, for large values of $n$ the adaptive procedure yields a significant improvement over the non-adaptive scheme, e.g., for $n=100$ we gain two orders of magnitude in error probability after 4 cycles of exploration. The improvement attained for small values of $n$ is also considerable, e.g., for $n=20$ it is one order of magnitude. It is noteworthy that the choices of $\epsilon_n=n^{-2/3}$ and $p_i=\gamma_n=n^{1/5}$ have been arbitrary and extensive simulations show that the gains are not very sensitive to these choices.

Next, we investigate the agility gain. Define the normalized sampling budget as the aggregate sampling budget divided by the number of channels. In Figure~\ref{fig:agility} we plot this quantity against the number of channels $n$, when requiring both approaches to have the same error probability for detecting $T=2$ holes. For the adaptive scheme we consider the performance with $K=1,\dots,5$ cycles of exploration. Again we consider the choices of $\epsilon_n=n^{-2/3}$ and $p_i=\gamma_n=n^{1/5}$. It is seen that the agility is improved by increasing $K$ and for $K=5$ the adaptive procedure requires about $80 \%$ less sampling budget than the non-adaptive procedure. In other words the adaptive procedure is 5 times faster than the non-adaptive scheme for detecting a hole with error probability is $P_{\tiny{\mbox{NA}}}=10^{-4}$

\begin{figure}[t]
\centering
\includegraphics[width=3.5in]{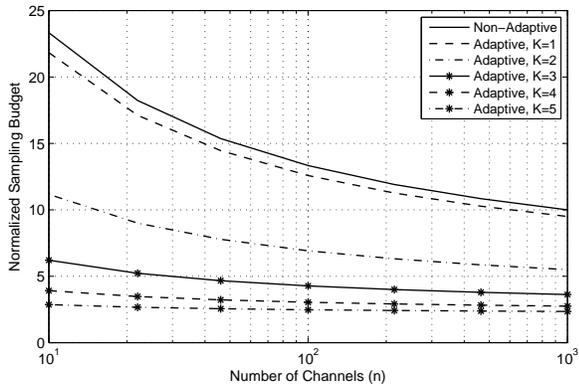}\\
\caption{Average number of observations taken per channel versus the number of channels $n$. All the curves correspond to the target error probability for detecting $T=2$ holes is $P_{\tiny{\mbox{A}}}(n)=P_{\tiny{\mbox{NA}}}(n)=10^{-4}$.}
\label{fig:agility}
\end{figure}

\begin{figure}[t]
\centering
\includegraphics[width=3.5in]{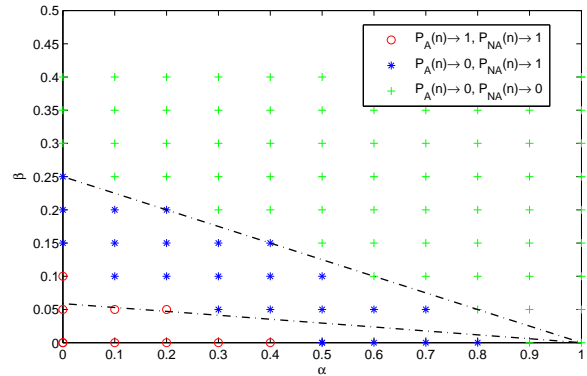}\\
\caption{Detectability regions for $M=5$, $K=4$, $\epsilon_n=n^{\alpha-1}$ and $\gamma_n=n^\beta$.}
\label{fig:detectability}
\end{figure}

\begin{figure}[!t]
\centering
\includegraphics[width=3.5in]{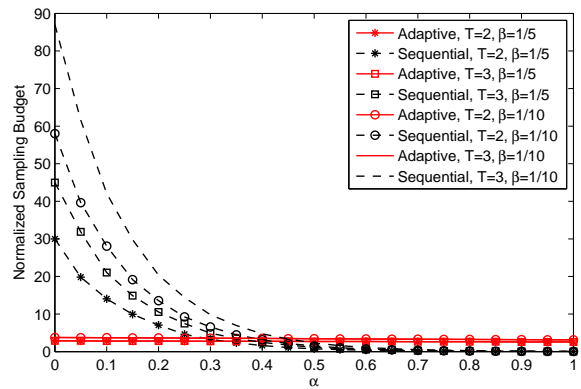}\\
\caption{Average number of observations taken per channel versus sparsity of the holes in the adaptive and the sequential method of \cite{Lai:IT10_Submitted} for achieving $P_{\tiny{\mbox{A}}}(n)=10^{-4}$.}
\label{fig:sequential}
\end{figure}

\begin{figure*}[t]
\hspace{-.5in}
\includegraphics[width=7.5in]{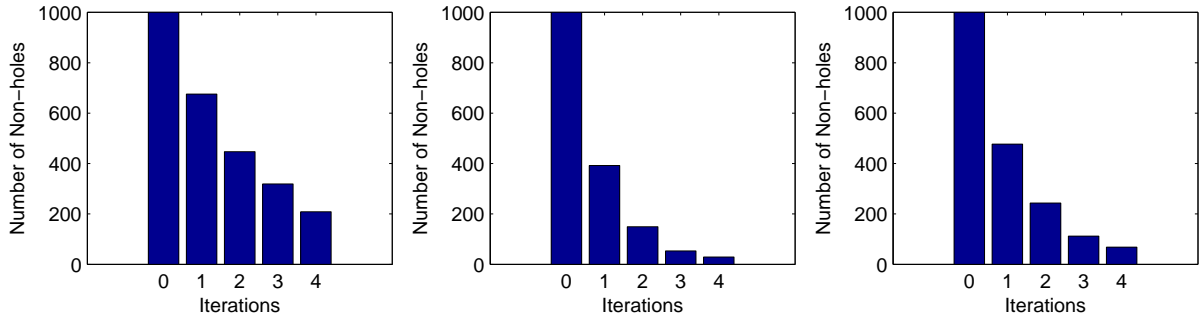}\\
\caption{Illustration of the exploration phase in eliminating the non-holes for $n=1000$ in three different system settings.}
\label{fig:histogram}
\end{figure*}

Figure~\ref{fig:detectability} depicts the results for assessing the necessary and sufficient conditions on the scaling of the power of the active users provided in Corollaries~\ref{cor:power_NA}~and~\ref{cor:power_adaptive} for non-adaptive and adaptive procedures. This is intended to depict that there exists a region where the non-adaptive procedure fails to identify spectrum holes, while the adaptive procedure succeeds. We assume that $\epsilon_n=n^{\alpha-1}$ for $\alpha\in[0,1]$ and $\gamma_n=n^\beta$ for $\beta>0$. By using Corollaries~\ref{cor:power_NA}~and~\ref{cor:power_adaptive} we find that the necessary and sufficient conditions for successful detection by the non-adaptive and adaptive procedures are $\beta\geq \frac{1-\alpha}{M}$ and $\beta\geq \frac{1-\alpha}{M'}$, respectively, where $M'\geq 2^K(M-2)+2$. The dashed lines depict these theoretical boundaries and split the space of $(\alpha,\beta)$ into three separate regions. We also provide the simulation results for some pairs of $(\alpha,\beta)$ in each of these regions. It is expected that in the lower region both procedures fail; in the middle region only the adaptive procedure succeeds and in the upper region both succeed. The simulation results (except for a few points) support the theoretical asymptotic predictions. In the simulations we have set $M=5$, $K=4$, and $n=1000$. We remark that requiring as low as about 5 sensing samples per channel is substantially less than the compressed-sensing-based approaches (e.g., \cite{Tan:ICASSP07}), that for obtaining a reliable estimate of the PSD take samples at a rate equal to at least half Nyquist rate, which for a wideband channel is large.

In Figure~\ref{fig:sequential}, we compare the agility of our proposed adaptive procedure with that of the quickest search~\cite{Lai:IT10_Submitted} for achieving the target error probability $P_{\tiny{\mbox{A}}}(n)=10^{-4}$. We compare the two schemes at different sparsity levels for the distribution of the spectrum holes. Specifically, we consider $\epsilon_n=n^{\alpha-1}$ for $\alpha\in[0,1]$ where smaller values of $\alpha$ correspond to sparser distribution of the holes. The objective is to identify $T=2,3$ holes out of $n=1000$ channels. The powers for the active users are assumed to be $\gamma=n^\beta$ for $\beta=1/5, 1/10$ and we deploy $K=4$ cycles of exploration. It is observed that for highly sparse distribution of holes (small enough values of $\alpha$,) our proposed adaptive procedure performs substantially better and for higher values of $\alpha$ which correspond to the case that there exist abundant number of spectrum holes, the quickest search method is outperforming.

Finally, Figure~\ref{fig:histogram} illustrates three snapshots showing the efficiency of the exploration phase in eliminating the non-holes for three different system realization. We consider $n=1000$ channels and set $\epsilon_n=n^{-2/3}$, $\gamma_n=n^{1/5}$. Figure~\ref{fig:histogram} plots the performance during $K=4$ cycles of exploration, where it shows that the number of non-holes reduces considerably after each exploration cycle. Also for these system realizations we have 10 holes, and we have observed that all the holes are retained throughout the exploration phase.
\section{Conclusions}
\label{sec:conclusions}
In this paper we have presented an adaptive sensing methodology for identifying multiple spectral holes in a wideband cognitive radio system. By gradually adjusting the measurement process using information gleaned from the previous measurements we are able to significantly improve the probability of correctly detecting multiple spectrum holes. This dramatic gain is patent both in the theoretical analysis and in the simulation results. More importantly, in the cognitive radio setting this improvement translates into a significant increase in agility, allowing spectrum holes to be identified quickly and therefore widening the window of time for opportunistic data transmission. Such dramatic improvements are not possible without the use of adaptive sampling techniques, as demonstrated.

\appendices
\begin{figure*}[b]
\hrulefill
\setcounter{MYtempeqncnt}{\value{equation}}
\setcounter{equation}{29}
\begin{align}
  \nonumber P(m_{k}\leq b)&= P\left(m_{k} \leq \frac{\gamma_n}{1+\gamma_n}\cdot m_{k-1}\right)\\
  & \label{eq:m_bound1} \leq \left[\Big((1-a_k)(1+\gamma_n)\Big)^{\frac{1}{1+\gamma_n}} \left(\frac{a_k(1+\gamma_n)}{\gamma_n}\right)^{\frac{\gamma_n}{1+\gamma_n}}\right]^{m_{k-1}} \\
  \label{eq:m_bound2} &  \leq \left[\Big((1+\gamma_n)\exp({-\lambda_k(1+\gamma_n)})\Big)^{\frac{1}{1+\gamma_n}} \left(\frac{1+\gamma_n}{\gamma_n}\right)^{\frac{\gamma_n}{1+\gamma_n}}\right]^{m_{k-1}} \\
  \nonumber &=\Bigg[\exp\Bigg(\frac{1}{1+\gamma_n}\Bigg(\log (1+\gamma_n)-\lambda_k(1+\gamma_n)+\underset{\leq 1}{\underbrace{\gamma_n\log\left(1+\frac{1}{\gamma_n}\right)}}\bigg)\Bigg)\Bigg]^{m_{k-1}}\\
   \nonumber &\leq \Bigg[\exp\Bigg(\Bigg(\frac{\log (1+\gamma_n)}{1+\gamma_n}-\lambda_k+\frac{1}{1+\gamma_n}\Bigg)\Bigg)\Bigg]^{m_{k-1}}\\
  \label{eq:m_bound3} &\leq \exp\Big(-\frac{m_{k-1}}{n^\alpha}\Big)\ ,
\end{align}
\setcounter{equation}{\value{MYtempeqncnt}}
\end{figure*}

\section{Proof of Lemma \ref{lmm:extreme}}\label{app:lmm:extreme}
Let $F(t;m)\dff P(Y_i\leq t)$ denote the CDF of the $\dGamma(M,\alpha_m)$ random variables $\{Y_i\}_{i=1}^m$. Also let $Q_{(i)}(w;m)\dff P(W^{(i)}_m\leq w)$ denote the CDF of $W^{(i)}_m$. Therefore,
\begin{align}\label{eq:lmm:extreme1}
\nonumber &Q_{(i)}(w;m) =  P(W^{(i)}_m\leq w)= P(Y^{(i)} \leq w b_m)\\
\nonumber  &= 1-\sum_{k=0}^{i-1}{m\choose k}\big[F(wb_m;m)\big]^k\big[1-F(wb_m;m)\big]^{m-k}\\
&= 1-\big[1-F(wb_m;m)\big]^{m} \ \sum_{k=0}^{i-1}{m\choose k}\left[\frac{F(wb_m;m)}{1-F(wb_m;m)}\right]^k\ .
\end{align}
Now recall that for a $\dGamma(M,\alpha_m)$ distribution
\begin{equation*}
    F(y;m)=\frac{1}{\Gamma(M)\alpha_m}\int_0^y \left(\frac{x}{\alpha_m}\right)^{M-1}\exp\left(-\frac{x}{\alpha_m}\right) dx\ .
\end{equation*}
Noting that for $0\leq x\leq y$ we have $\exp(-y/\alpha_m)\leq \exp(-x/\alpha_m)\leq 1$ we immediately get
\begin{equation*}
    \frac{\exp(-y/\alpha_m)}{\Gamma(M+1)}\cdot \left(\frac{y}{\alpha_m}\right)^M  \ \leq\ F(y;m)\ \leq\frac{1}{\Gamma(M+1)}\cdot \left(\frac{y}{\alpha_m}\right)^M\ .
\end{equation*}
By noting that we had defined  $b_m\dff\alpha_m\left[\frac{\Gamma(M+1)}{m}\right]^{\frac{1}{M}}$, after some simplifications we find the following bounds above for $F(wb_m;m)$
\begin{equation*}
    \frac{w^M}{m}\cdot \exp\left(-w^M\cdot \frac{\Gamma(M+1)}{m}\right) \ \leq \ F(wb_m;m) \ \leq \ \frac{w^M}{m} \ ,
\end{equation*}
By using the \emph{little-o} notation as $m\rightarrow\infty$ we have $F(wb_m;m)=\frac{w^M}{m} (1+o(1))$ . Hence, for the term $\left[1-F(wb_m;m)\right]^m$ as $m\rightarrow\infty$ we have
\begin{align}\label{eq:lmm:extreme2}
\nonumber\left[1-F(wb_m;m)\right]^m & = \left[1-\frac{w^M}{m} (1+o(1))\right]^m\\
\nonumber &= \nonumber  \left[1-\frac{w^M}{m}+o(m^{-1})\right]^m\\
\nonumber &= \exp\left[m \log \left(1-\frac{w^M}{m}+o(m^{-1})\right)\right]\\
\nonumber &= \exp\left[-w^M+o(1)\right]\\
& \doteq \exp(-w^M)\ .
\end{align}
On the other hand for every fixed $k$ we get
\begin{align}\label{eq:lmm:extreme3}
\nonumber {m\choose k} & \left[\frac{F(wb_m;m)}{1-F(wb_m;m)}\right]^k\\
\nonumber & =\frac{\prod_{i=0}^{k-1}m(1-i/m)}{k!} \left[\frac{\frac{w^M}{m} (1+o(1))}{1-\frac{w^M}{m}(1+o(1))}\right]^k\\
\nonumber &= \frac{\prod_{i=0}^{k-1}m(1-o(1))}{k!} \left[\frac{w^M}{m} (1+o(1))\right]^k\\
& \doteq \frac{m^k}{k!}\left[\frac{w^M}{m}\right]^k=\frac{w^{Mk}}{k!}\ .
\end{align}
Equations \eqref{eq:lmm:extreme1}-\eqref{eq:lmm:extreme3} yield
\begin{equation*}
    Q_{(i)}(w;m)\doteq 1-\exp(-w^M)\sum_{k=0}^{i-1}\frac{w^{KM}}{k!}\ ,
\end{equation*}
which is the desired result.

\section{Proof of Lemma \ref{lmm:hole}}
\label{app:lmm:hole}
For the non-adaptive procedure by invoking the assumption $\epsilon_n=o(1)$, we know that $\gamma_n=\omega(1)$ from Corollary~\ref{cor:power_NA}. Throughout the analysis for the adaptive procedure we also assume that $\gamma_n=\omega(1)$ and characterize the precise necessary and sufficient condition for the scaling rate of $\gamma_n$ at the end of Section~\ref{sec:adaptive}. We start with the following remarks.
\begin{remark}\label{remark:gamma}
Let $X\sim \dGamma(M,a)$. Then for any $\lambda>0$ we have $P(X< \lambda)\geq 1-\exp\left(-\frac{\lambda}{a}\right)$.
\end{remark}
\begin{remark}\label{remark:chernoff}
For $B$ distributed as ${\rm Binomial}(m,a)$ and $b<\bbe[B]$ we have
\begin{equation*}
    P(B\leq b)\leq\left(\frac{m-ma}{m-b}\right)^{m-b}\left(\frac{ma}{b}\right)^b.
\end{equation*}
\end{remark}
Let us define $\J_0=\G_0\cap\H_0$ and for $k=1,\dots,K$ let $\J_k\dff\G_k\cap\H_0$ be the set containing the indices of the spectrum holes retained by the $k^{th}$ exploration cycle. By the definition of $m_k$ we have $|\J_k|=m_k$. Also, for any $i\in\G_{k-1}$ define
\begin{equation*}
    T^k_i\dff\boldsymbol{1}\{Y^k_i<\lambda_k(1+\gamma_n)\}\ ,
\end{equation*}
where
\begin{equation*}
\boldsymbol{1}\{{\cal A}\}=\left\{
    \begin{array}{cc}
      1 & \mbox{if}\;\; {\cal A}\;\mbox{is true}\\
      0 & \mbox{if}\;\; {\cal A}\;\mbox{is false}
    \end{array}\right., \quad\mbox{for}\quad k=1,\dots,K\ .
\end{equation*}
According to the algorithm in Table 1 for the $i^{th}$ channel, which can be either a hole or a non-hole, $T^k_i$ is 1 if is retained by the $k^{th}$ exploration cycle and is 0 otherwise. Therefore, the number of {\em holes} retained by the $k^{th}$ exploration cycle is $m_{k}=\sum_{i\in\J_{k-1}}T^{k}_i$. By noting that $|\J_{k-1}|=m_{k-1}$ we find that
\begin{equation*}
     m_{k}\sim{\rm Binomial}(m_{k-1},a_k),\quad\mbox{for}\quad k=1,\dots,K\ ,
\end{equation*}
where
\begin{align*}
    a_k & \dff P(T^k_i=1\med i\in\J_{k-1})\\
    & =P(Y^k_i<\lambda_k(1+\gamma_n)\med i\in\J_{k-1})\leq 1\ .
\end{align*}
Based on (\ref{eq:Y_dist_adaptive}), $Y^k_i$ is distributed as $\dGamma(M_k,1)$ for $i\in\J_{k-1}$. Therefore, by using Remark \ref{remark:gamma} we find that
\begin{equation}
    \label{eq:Gamma_bound}
    a_k\geq 1-\exp(-\lambda_k(1+\gamma_n)) \ .
\end{equation}
\begin{align*}
    P&\left(\frac{\gamma_n}{1+\gamma_n}\cdot m_{k-1}\leq m_{k} \leq  m_{k-1}\right)\\
    &  =P\left(m_{k} \leq \frac{\gamma_n}{1+\gamma_n}\cdot m_{k-1}\right)\\
    & \geq 1-\exp\Big(-\frac{m_{k-1}}{n^\alpha}\Big)\ ,
\end{align*}
which concludes the result.
\begin{figure*}[b]
\setcounter{MYtempeqncnt}{\value{equation}}
\setcounter{equation}{36}
\hrulefill
\begin{equation}\label{eq:A_event}
A_k=\left\{m_{k-1}\geq m_{k} \geq \Big(1-\frac{1}{1+\gamma_n}\Big)m_{k-1}\quad\mbox{ and } \quad \Big(\frac{1}{2}-\frac{1}{\log n}\Big)\ell_{k-1} \leq  \ell_{k}\leq  \Big(\frac{1}{2}+\frac{1}{\log n}\Big)\ell_{k-1}\right\}\ .
\end{equation}
\hrulefill
\begin{align}\label{eq:A_event2}
    A\dff\left\{m_{0}\geq m_{K} \geq \Big(1-\frac{1}{1+\gamma_n}\Big)^Km_0\quad\mbox{and}\quad \Big(\frac{1}{2}-\frac{1}{\log n}\Big)^K\ell_0 \leq  \ell_{K}\leq  \Big(\frac{1}{2}+\frac{1}{\log n}\Big)^K\ell_0\right\}\ .
\end{align}
\setcounter{equation}{\value{MYtempeqncnt}}
\end{figure*}

For sufficiently large $n$ (and thereof sufficiently large $\gamma_n$) we have $\frac{1}{1+\gamma_n}\geq \exp(-\lambda_k(1+\gamma_n))$ which in conjunction with (\ref{eq:Gamma_bound}) provides that for sufficiently large $n$
\begin{equation*}
    \frac{\gamma_n}{1+\gamma_n}\cdot m_{k-1}\leq a_km_{k-1} \ .
\end{equation*}
Therefore, by setting $b\dff\frac{\gamma_n}{1+\gamma_n}\cdot m_{k-1}$, for the random variable $m_{k}\sim{\rm Binomial}(m_{k-1},a_k)$ we have $b\leq \bbe[m_{k}]= a_km_{k-1}$. As a result, for this choice of $b$ the condition of Remark \ref{remark:chernoff} is satisfied and, for sufficiently large $n$ we find that $P(m_{k}\leq b)$ is lower bounded according to \eqref{eq:m_bound3}. Note that (\ref{eq:m_bound1}) holds according to Remark~\ref{remark:chernoff}. (\ref{eq:m_bound2}) is obtained by invoking the inequality given in (\ref{eq:Gamma_bound}), and noting that $a_k\leq 1$ as $a_k$ is a probability. Finally the inequality in (\ref{eq:m_bound3}) is justified by noting that for sufficiently large $n$ and any $\alpha>0$
\begin{equation*}
    \frac{\log (1+\gamma_n)}{1+\gamma_n}-\lambda_k+\frac{1}{1+\gamma_n}\leq -\frac{1}{n^\alpha} ,
\end{equation*}
as all terms in the LHS and RHS except the constant $-\lambda_k$ are diminishing in the asymptote of large $n$. Recall that $\J_{k}\subseteq\J_{k-1}$ and consequently, $m_{k}\leq m_{k-1}$. This fact in conjunction with \eqref{eq:m_bound3} provides that
\section{Proof of Lemma \ref{lmm:occupied}}
\label{app:lmm:occupied}
Let $\L_0=\G_0\cap\H_1$ and for $k=1,\dots,K$ define $\L_k\dff\G_k\cap\H_1$ which contains the set of the indices of the non-holes contained retained by the $k^{th}$ exploration cycle. By the definition of $\ell_k$ we have and $|\L_k|=\ell_k$. By recalling the definition of $T^k_i$ and taking into account their independence, the number of non-holes that the $k^{th}$ exploration cycle retains is $\ell_{k}=\sum_{i\in \L_{k-1}}T^k_i$. Applying Hoeffding's inequality provides
\setcounter{equation}{32}
\begin{equation}
  \label{eq:lmm:hole:proof1}
  P\bigg(\big|\ell_{k}-\bbe\big[\ell_{k}]\big|>\frac{\ell_{k-1}}{\log n}\bigg)\leq 2\exp\left(-\frac{2\ell_{k-1}}{(\log n)^2}\right).
\end{equation}
On the other hand
\setcounter{equation}{33}
\begin{equation}\label{eq:lmm:hole:proof2}
    \bbe\big[\ell_{k}]=\sum_{i\in \L_{k-1}}\bbe[T^{k}_i]=\sum_{i\in \L_{k-1}}P(T^{k}_i=1\med i\in\L_{k-1})\ .
\end{equation}
According to (\ref{eq:Y_dist_adaptive}), for $i\in\L_{k-1}$, $Y^k_i$ is distributed as $\dGamma(M_k,1+p_i)$. For any $i\in\L_{k-1}$ we denote the CDF of ${\rm Gamma}(M_k,1+p_i)$ by $F^k_i(t;n)$. Since $p_i\geq \gamma_n$ we get
\begin{align*}
    P& (T^k_i=1 \med i\in\L_{k-1}) =F^k_i\Big(\lambda_k(1+\gamma_n)\Big)\\
    & \leq F^k_i\Big(\lambda_k(1+p_i)\Big)\\
    &= \frac{1}{\Gamma(M_k)(1+p_i)^M}\int_0^{\lambda_k(1+p_i)}x^{M-1}\exp\left(-\frac{x}{1+p_i}\right)dx\\
    & = \frac{1}{\Gamma(M_k)}\int_0^{\lambda_k}x^{M-1}\exp\left(-x\right)dx\\
    & =\frac{1}{2},
\end{align*}
where the last step holds as we had defined $\lambda_k$ as the median of ${\rm Gamma}(M_k,1)$. As a result, from \eqref{eq:lmm:hole:proof2} we find that  $\bbe\big[\ell_{k}]\leq\frac{|\L_{k-1}|}{2}=\frac{\ell_{k-1}}{2}$. By selecting $c_k\dff\frac{\bbe[\ell_{k}]}{\ell_{k-1}}\leq\frac{1}{2}$ and recalling (\ref{eq:lmm:hole:proof1}) we get the desired result. Finally note that $F^k_i\Big(\lambda_k(1+\gamma_n)\Big)$ is monotonic in $p_i$. Therefore, the necessary and sufficient condition for having $\bbe\big[\ell_{k}]=\frac{\ell_{k-1}}{2}$ is that $F^k_i(\lambda_k(1+\gamma_n))= F^k_i(\lambda_k(1+p_i))$, which in turn requires that $p_i=\gamma_n$ for $i\in\H_1$.

\section{Proof of Theorem \ref{thm:adaptive}}\label{app:thm:adaptive}
Similar to the proof of Theorem \ref{thm:NA_bound}, for quantifying the performance of the {\em robust} detection scheme, we consider the worst-case scenario and set  $p_i=\gamma_n$ for all $i\in\H_1$. Define the event $A_k$ as in \eqref{eq:A_event}. 

Lemma~\ref{lmm:hole} and ~\ref{lmm:occupied} characterize the probability of event $A_k$ given $m_{k-1}$ and $\ell_{k-1}$, and state that,
\begin{align}\label{eq:zeta}
    \nonumber \zeta_{k}&\dff P(A_k\med m_{k-1},\ell_{k-1})\\
    \nonumber &\geq P\left(m_{k-1}\geq m_{k} \geq \Big(1-\frac{1}{1+\gamma_n}\Big)m_{k-1}\med m_{k-1}\right)\\
    \nonumber & + P\left(\Big(\frac{1}{2}-\frac{1}{\log n}\Big)\ell_{k-1} \leq  \ell_{k}\leq  \Big(\frac{1}{2}+\frac{1}{\log n}\Big)\ell_{k-1}\med \ell_{k-1}\right)\\
    & \geq 1-\exp(-\frac{m_{k-1}}{n^\alpha})-2\exp(-\frac{2\ell_{k-1}}{(\log n)^2})\ .
\end{align}
By concatenating all the bounds we conclude that the event $(A_1,A_2,\dots,A_K)$ is a sufficient condition for the event $A$ defined in \eqref{eq:A_event2}. 

By noting that given $m_{k-1}$ and $\ell_{k-1}$, the event $A_k$ is independent of events $A_i$ for $i<k$, we immediately conclude that the event $A$ holds with probability at least
\begin{align*}
    P(A)&\geq P(A_1,A_2,\ldots,A_K)\\
    & =\prod_{i=1}^{K}P(A_k\med m_{k-1},\ell_{k-1})=\prod_{i=1}^{K}\zeta_i\ .
\end{align*}
In the asymptote of large $n$, for any {\em fixed} $K$ we have (conditions on the choice of $K$ are discussed in Theorem~\ref{th:K})
\begin{align}\label{eq:1}
\nonumber    &\lim_{n\rightarrow\infty}\Big(1-\frac{1}{1+\gamma_n}\Big)^K=1\ ,\\
\nonumber    \mbox{and}\quad &\lim_{n\rightarrow\infty}\Big(\frac{1}{2}-\frac{1}{\log n}\Big)^K=1\ ,\\
\mbox{and}\quad & \lim_{n\rightarrow\infty}\Big(\frac{1}{2}+\frac{1}{\log n}\Big)^K=1\ ,
\end{align}
which provide that in the asymptote of large $n$
\begin{equation*}
    A=\left\{m_K=m_0\quad\mbox{and}\quad \ell_K=\frac{\ell_0}{2^K}\right\}
\end{equation*}
with probability at least $\prod_{i=1}^{K}\zeta_i$. On the other hand, by recalling the assumptions $\epsilon_n=o(1)$ and $n\epsilon_n=\omega(1)$, according to the the law of large numbers we get $\ell_0\doteq n(1-\epsilon_n)\doteq n$ and $m_0\doteq n\epsilon_n=\omega(1)$. Hence, we can ensure that for $k=1,\dots,K$, $\zeta_k\rightarrow 1$ as $n\rightarrow\infty$. More specifically, by selecting $\alpha<\lim_{n\rightarrow\infty}\frac{\log n\epsilon_n}{\log n}$ for $\zeta_1$ we have
\begin{equation*}
    \lim_{n\rightarrow\infty}\frac{m_0}{n^\alpha}= \lim_{n\rightarrow\infty}\frac{n\epsilon_n}{n^\alpha}=+\infty\ ,
\end{equation*}
which shows that
\setcounter{equation}{38}
\begin{equation}\label{eq:exp1}
     \lim_{n\rightarrow\infty}\exp(-\frac{m_{0}}{n^\alpha})=0\ ,
\end{equation}
and
\begin{equation*}
    \lim_{n\rightarrow\infty}\frac{2\ell_{0}}{(\log n)^2} = \lim_{n\rightarrow\infty}\frac{2n}{(\log n)^2}=+\infty\ ,
\end{equation*}
which shows that
\begin{equation}\label{eq:exp2}
 \lim_{n\rightarrow\infty}2\exp\left(-\frac{2\ell_{k-1}}{(\log n)^2}\right)=0\ .
\end{equation}
Equations \eqref{eq:exp1}-\eqref{eq:exp2} and \eqref{eq:zeta} yield that $\lim_{n\rightarrow\infty}\zeta_1=1$. By following the same procedure, the same result for all $\zeta_k$ can be recovered. As a result, for sufficiently large $n$, with probability 1 we have $m_{K}=m_0$ and $\ell_{K}=\ell_0/2^K$.

Now let us define $\tilde n_0$ and $\tilde n_1$ as the number of holes and non-holes retained after the exploration phase. By following the same line of argument as in the proof of Theorem~\ref{thm:NA_bound} and by appropriately applying Lemma~\ref{lmm:extreme} we can show that
\begin{equation}\label{eq:thm:adaptive1}
P_{\mbox{\tiny A}}(n)\doteq 1- \left(1+[(1+\gamma_n)^M\cdot\frac{\tilde n_0}{\tilde n_1}\epsilon_n]^{-1}\right)^{-T}\ .
\end{equation}
Note that for sufficiently large $n$
\begin{equation*}
    \frac{\tilde n_0}{\tilde n_1}\rightarrow\frac{m_K}{\ell_K}\rightarrow \frac{m_0}{\ell_0}\cdot 2^K=\frac{n_0}{n_1}\cdot 2^K\rightarrow 2^K\epsilon_n\ .
\end{equation*}
Plugging this asymptotic value of $\frac{\tilde n_0}{n_1}$ into \eqref{eq:thm:adaptive1} completes the proof.

\section{Proof of Theorem \ref{th:agility}}
\label{app:th:agility}
Based on the structure of the experimental design, for a given $Mn$ sampling budget the non-adaptive procedure takes $M$ observations per channel. By equating the error probabilities of hole detection in the non-adaptive and adaptive schemes given in (\ref{eq:P_NA}) and (\ref{eq:P_adaptive}) we find that as $n\rightarrow\infty$
\begin{equation*}
    (1+\gamma_n)^{M} \epsilon_n\doteq (1+\gamma_n)^{M_{K+1}} \cdot 2^K\epsilon_n\ ,
\end{equation*}
which shows that
\begin{equation}
    \label{eq:M_K} M_{K+1}\doteq M-\frac{K}{\log(1+\gamma_n)}\leq M\ .
\end{equation}
Next, let us evaluate the sampling budget required by the non-adaptive procedure given that in the detection phase we take $M_{K+1}$ samples per channel that has survived the exploration phase. As discussed earlier, the characterization of $P_{\tiny{\mbox{A}}}(n)$ conveys that the optimal experimental design is to take 1 sample per channel in each exploration cycle. Therefore, the aggregate amount of samples taken throughout exploration and detection phases is $\sum_{k=0}^{K}(m_k+\ell_k)$. On other hand, from lemmas~\ref{lmm:hole} and \ref{lmm:occupied} we find that in the asymptote of large $n$ we have $m_k=m_0$ and $\ell_k=\frac{1}{2^{k}}\ell_1$ with probability 1 for $k=1,\dots,K$. Recalling that $\frac{m_0}{\ell_0}\rightarrow\epsilon_n\rightarrow 0$ and $\ell_0\doteq n$ we find that the sampling budget required by the adaptive procedure is {\em asymptotically equal} to
\begin{align}\label{eq:M_K1}
    \nonumber \sum_{k=0}^{K-1}\ell_k+\ell_KM_{K+1} & \doteq \sum_{k=0}^{K-1}\frac{n}{2^{k}}+\frac{n}{2^K}\cdot M_{K+1}\\
    \nonumber &=n\left(2-\frac{1}{2^{K-1}}+\frac{M_{K+1}}{2^K}\right)\\
    & \leq n\left(2+\frac{M_{K+1}}{2^K}\right)\ .
\end{align}
Equations \eqref{eq:M_K} and \eqref{eq:M_K1} together show that the sampling budget of the adaptive procedure is asymptotically upper bounded by $n(2+M/{2^K})$. As a result, for maintaining $P_{\tiny{\mbox{A}}}(n)=P_{\tiny{\mbox{NA}}}(n)$, the asymptotic ratio of the sampling budget required by the non-adaptive procedure to that required by the adaptive scheme is lower bounded by
\begin{equation*}
    \frac{Mn}{n(2+M/{2^K})}=\left(\frac{2}{M}+\frac{1}{2^K}\right)^{-1}\ .
\end{equation*}
\section{Proof of Corollary \ref{cor:power_adaptive}}
\label{app:cor:power_adaptive}
First we find how $M_K$ and $M$ should be related such that they give rise to the same sampling budget required by the adaptive and non-adaptive schemes. By following a similar line of argument as in the proof of Theorem \ref{th:agility} and \eqref{eq:M_K1} we find that
\begin{align*}
    \nonumber Mn =\sum_{k=0}^{K}(m_k+\ell_k) \leq  n\Bigg(\frac{2^{K}-1}{2^{K-1}}+\frac{M_{K+1}}{2^{K}}\Bigg),
\end{align*}
which provides $M_{K+1}\geq 2^{K}(M-2)+2$. On the other hand, from \eqref{eq:P_adaptive} forcing $P_{\mbox{\tiny A}}(n)\rightarrow 0$ yields that it is necessary and sufficient to have
\begin{equation*}
    \gamma_n= \omega\left(\sqrt[M_{K+1}]{\frac{1}{2^K\epsilon_n}}\right)\ ,
\end{equation*}
Setting $M'=M_{K+1}$ establishes the desired result.

\section{Proof of Theorem \ref{th:K}}
\label{app:th:K}
A key constraint that guarantees retaining almost all the holes throughout the exploration phase is cast in \eqref{eq:1} as
\begin{equation}\label{eq:1_1}
    \lim_{n\rightarrow\infty}\Big(1-\frac{1}{1+\gamma_n}\Big)^K=1\ .
\end{equation}
According to Corollary~\ref{cor:power_adaptive}, guaranteeing this equality necessitates that $K$ does not grow faster than $\gamma_n$, i.e.,
\begin{equation}\label{eq:1_2}
K\in o\left(\gamma_n\right)\quad\mbox{while}\quad \gamma_n= \omega\left(\sqrt[M']{\frac{1}{2^K\epsilon_n}}\right)\  .
\end{equation}
Define
\begin{equation}\label{eq:1_3}
\Omega_n\dff\{\gamma_n\med \gamma_n\mbox{ does {\bf not} satisfy \eqref{eq:gamma_adaptive}}\}\ .
\end{equation}
From \eqref{eq:1_2}-\eqref{eq:1_3} we obtain that
\begin{equation}\label{eq:1_4}
K=\sup_{\Omega_n}\gamma_n\doteq \sqrt[M']{\frac{1}{2^K\epsilon_n}}\quad\Rightarrow\quad M'\log K+2\doteq \log \frac{1}{\epsilon_n}\ .
\end{equation}
By taking into account the characterization of $M'$ in Corollary~\ref{cor:power_adaptive} we find the following upper bound on $K$.
\begin{equation}\label{eq:1_5}
(2^K(M-2)+2)\log K+2\dotlt\log \frac{1}{\epsilon_n}\ .
\end{equation}
Next, we identify the scaling law of $K$ with respect to $n$ such that the asymptotic inequality above is satisfied. As $K$ scales, the constant summands (i.e., +2) will be dominated by other summands, the inequality above can be stated as
\begin{equation}\label{eq:1_6}
\log(M-2)+K+\log\log K\dotlt\log\log\frac{1}{\epsilon_n}\ .
\end{equation}
$(M-2)$ being a constant as well as $\log\log K$ are both dominated by $K$, giving rise to the following upper bound on $K$
\begin{equation}\label{eq:1_7}
K\dotlt\log\log\frac{1}{\epsilon_n}\ .
\end{equation}

We also need to examine the constraints on $K$ that guarantee that the exploration phase discards the non-holes as expected in the proof of Theorem~\ref{thm:adaptive}. More specifically, we need to ensure that the two following equalities hold.
\begin{equation}\label{eq:2}
    \lim_{n\rightarrow\infty}\Big(\frac{1}{2}-\frac{1}{\log n}\Big)^K=1\quad\mbox{and}\quad \lim_{n\rightarrow\infty}\Big(\frac{1}{2}+\frac{1}{\log n}\Big)^K=1\ ,
\end{equation}
which indicate that $K$ must satisfy
\begin{equation}\label{eq:2_1}
K\in o(\log n)\ .
\end{equation}
By recalling the assumption $n\epsilon_n=\omega(1)$, or its equivalent form $\frac{1}{\epsilon_n}=o(n)$, it can be readily verified that the bound given in \eqref{eq:1_7} is stronger than that in \eqref{eq:2_1}. By taking into account the monotonicity of the detection performance with increasing $K$, the optimal choice of $K$ is the upper bound in \eqref{eq:1_7}.


\section{Proof of Remark \ref{rem:MAP}}
\label{app:rem:MA}
The probability that a fixed set of channels ${\cal U}$ with cardinality $T$ contains only holes is
\begin{align}
\nonumber & P({\cal U}\subseteq \H_0\med \D_n) =  \prod_{i\in {\cal U}}P(i\in \H_0\med \D_n)\\
\nonumber &=\prod_{i\in {\cal U}}\frac{P(i\in \H_0)\cdot P(\D_n\med i\in \H_0)}{P(\D_n)}\\
\nonumber & = \prod_{i\in {\cal U}}\frac{P(i\in\H_0)\cdot P\big(\bX_i\med i\in \H_0\big)}{P\big(\bX_i\big)}\\
\nonumber &=\prod_{i\in {\cal U}}  \frac{\frac{\epsilon_n}{\pi^M} \exp{\left(-\|\bX_i\|^2\right)}}{\frac{\epsilon_n}{\pi^M} \exp{\left(-\|\bX_i\|^2\right)}+ \frac{1-\epsilon_n}{(\pi(1+p_i))^M}\exp{\left(-\frac{1}{1+p_i}\|\bX_i\|^2\right)}}\\
\label{eq:baysian}  &=\prod_{i\in {\cal U}}\left[1+\frac{1-\epsilon_n}{\epsilon_n} \cdot\frac{1}{(1+p_i)^M}\exp{\left(\frac{p_i}{1+p_i}\|\bX_i\|^2\right)}\right]^{-1}\ ,
\end{align}
where the above expansion follows from Bayes's rule and the statistical independence of the observations made over different channels. Hence, the MAP detector that maximizes $P({\cal U}\subseteq \H_0\med \D_n)$ among all sets of channels with cardinality $T$ is given by
\begin{eqnarray}\label{eq:MAP}
\nonumber \widehat {\cal U}_{\tiny{\mbox{MAP}}}^{\tiny{\mbox{NA}}} &\dff& \arg\max_{{\cal U}:\ |{\cal U}|=T}\log P({\cal U}\subseteq \H_0\med \D_n) \\
&=&\arg\min_{{\cal U}:\ |{\cal U}|=T}\sum_{i\in {\cal U}}\log\left[1+\frac{1-\epsilon_n}{\epsilon_n} \cdot U_i\right]\ .
\end{eqnarray}
where
\begin{equation}\label{eq:U}
    U_i=\frac{1}{(1+p_i)^M}\exp{\left(\frac{p_i}{1+p_i}\|\bX_i\|^2\right)}\ .
\end{equation}
Therefore, $\widehat {\cal U}_{\tiny{\mbox{MAP}}}^{\tiny{\mbox{NA}}}$ contains the indices of the $T$ smallest elements of the set $\{U_1,\dots,U_n\}$.

{\bibliographystyle{IEEEtran} \bibliography{IEEEabrv,CR_ASA}}

\begin{thebibliography}{10}
\providecommand{\url}[1]{#1}
\csname url@samestyle\endcsname
\providecommand{\newblock}{\relax}
\providecommand{\bibinfo}[2]{#2}
\providecommand{\BIBentrySTDinterwordspacing}{\spaceskip=0pt\relax}
\providecommand{\BIBentryALTinterwordstretchfactor}{4}
\providecommand{\BIBentryALTinterwordspacing}{\spaceskip=\fontdimen2\font plus
\BIBentryALTinterwordstretchfactor\fontdimen3\font minus
  \fontdimen4\font\relax}
\providecommand{\BIBforeignlanguage}[2]{{%
\expandafter\ifx\csname l@#1\endcsname\relax
\typeout{** WARNING: IEEEtran.bst: No hyphenation pattern has been}%
\typeout{** loaded for the language `#1'. Using the pattern for}%
\typeout{** the default language instead.}%
\else
\language=\csname l@#1\endcsname
\fi
#2}}
\providecommand{\BIBdecl}{\relax}
\BIBdecl

\bibitem{Giannakis:SP10}
J.~A. Bazerque and G.~B. Giannakis, ``Distributed spectrum sensing for
  cognitive radio networks by exploiting sparsity,'' \emph{{IEEE} Trans. Signal
  Process.}, vol.~58, no.~3, pp. 1847 -- 1862, Mar. 2010.

\bibitem{Han:11}
J.~Meng, W.~Yin, H.~Li, E.~Hossain, and Z.~Han, ``Collaborative spectrum
  sensing from sparse observations in cognitive radio networks,'' \emph{{IEEE}
  J. Sel. Areas Commun.}, vol.~29, no.~2, pp. 327 -- 337, Feb. 2011.

\bibitem{Haykin:09}
S.~Haykin, D.~J. Thomson, and J.~H. Reed, ``Spectrum sensing for cognitive
  radio,'' \emph{Proc. {IEEE}}, vol.~97, no.~5, pp. 849--877, May 2009.

\bibitem{Ma:09}
J.~Ma, G.~Y. Li, and B.~H. Juang, ``Signal processing in cognitive radio,''
  \emph{Proc. {IEEE}}, vol.~97, no.~5, pp. 805--823, May 2009.

\bibitem{Yucek:09}
T.~Yucek and H.~Arslan, ``A survey of spectrum sensing algorithms for cognitive
  radio applications,'' \emph{{IEEE} Commun. Surveys Tuts.}, vol.~11, no.~1,
  pp. 116--130, 2009.

\bibitem{Akyildiz:06}
F.~Akyildiz, W.-Y. Lee, M.~C. Vuran, and S.~Mohanty, ``Next generation dynamic
  spectrum access cognitive radio wireless networks: a survey,'' \emph{\it
  Computer Networks}, vol.~50, no.~13, pp. 2127--2159, Sep. 2006.

\bibitem{Lai:IT10_Submitted}
L.~Lai, H.~V. Poor, Y.~Xin, and G.~Georgiadis, ``Quickest search over multiple
  sequences,'' \emph{{IEEE} Trans. Inf. Theory}, vol.~57, no.~8, pp.
  5375--5386, Aug. 2011.

\bibitem{Quan:SPM08}
Z.~Quan, S.~Cui, H.~V. Poor, and A.~H. Sayed, ``Collaborative wideband sensing
  for cognitive radios,'' \emph{{IEEE} Signal Process. Mag.}, vol.~25, no.~6,
  pp. 60--73, Nov. 2008.

\bibitem{Malloy:ISIT11}
M.~Malloy and R.~Nowak, ``Sequential analysis in high-dimensional multiple
  testing and sparse recovery,'' in \emph{IEEE International Symposium on
  Information Theory (ISIT)}, Aug. 2011.

\bibitem{Tan:ICASSP07}
Z.~Tan and G.~B. Giannakis, ``Compressed sensing for wideband cognitive
  radios,'' in \emph{IEEE International Conference on Acoustics, Speech and
  Signal Processing (ICASSP)}, Honolulu, HI, Apr. 2007, pp. 1357--1360.

\bibitem{Leus:ICASSP09}
Y.~L. Polo, Y.~Wang, A.~Pandharipande, and G.~Leus, ``Compressive wide-band
  spectrum sensing,'' in \emph{IEEE International Conference on Acoustics,
  Speech and Signal Processing (ICASSP)}, Taipei, Taiwan, Apr. 2009, pp.
  2337--2340.

\bibitem{Wang:GC10}
Y.~Wang, Z.~Tian, and C.~Feng, ``A two-step compressed spectrum sensing scheme
  for wideband cognitive radios,'' in \emph{IEEE Global Telecommunications
  Conference}, Miami, FL, Dec. 2010.

\bibitem{Vilar:ICASSP10}
R.~Vazquez-Vilar, G.and L\'{o}pez-Valcarce, C.~Mosquera, and
  N.~Gonz\'{a}lez-Prelcic, ``Wideband spectral estimation from compressed
  measurements exploiting spectral a priori information in cognitive radio
  systems,'' in \emph{IEEE International Conference on Acoustics, Speech and
  Signal Processing (ICASSP)}, Dallas, TX, Apr. 2010, pp. 2958--2961.

\bibitem{Hong:GC10}
S.~Hong, ``Multi-resolution bayesian compressive sensing for cognitive radio
  primary user detection,'' in \emph{IEEE Global Telecommunications
  Conference}, Miami, FL, Dec. 2010.

\bibitem{Yu:ICASSP08}
Z.~Yu, S.~Hoyos, and B.~M. Sadler, ``Mixed-signal parallel compressed sensing
  and reception for cognitive radio,'' in \emph{IEEE International Conference
  on Acoustics, Speech and Signal Processing (ICASSP)}, Las Vegas, NV, Apr.
  2008, pp. 3861--3864.

\bibitem{Tian:GC08}
Z.~Tian, ``Compressed wideband sensing in cooperative cognitive radio
  networks,'' in \emph{IEEE Global Telecommunications Conference}, New Orleans,
  LA, Dec. 2008.

\bibitem{Angelosante:ICASSP09}
D.~Angelosante and G.~B. Giannakis, ``{RLS}-weighted lasso for adaptive
  estimation of sparse signals,'' in \emph{IEEE International Conference on
  Acoustics, Speech and Signal Processing (ICASSP)}, Taipei, Taiwan, Apr. 2009,
  pp. 3245--3248.

\bibitem{Angelosante:DSP09}
D.~Angelosante, G.~B. Giannakis, and E.~Grossi, ``Compressed sensing of
  time-varying signals,'' in \emph{IEEE International Conference on Digital
  Signal Processing (DSP)}, Santorini, Greece, Jul. 2009.

\bibitem{Wang:ITA09}
Y.~Wang, A.~Pandharipande, Y.~L. Poloy, and G.~Leus, ``Distributed compressive
  wide-band spectrum sensing,'' in \emph{Information Theory and Applications
  Workshop}, San Diego, CA, Feb. 2009.

\bibitem{Zeng:ICC10}
F.~Zeng, Z.~Tian, and C.~Li, ``Distributed compressive wideband spectrum
  sensing in cooperative multi-hop cognitive networks,'' in \emph{IEEE
  International Conference on Communications (ICC)}, Cape Town, South Africa,
  Jul. 2010.

\bibitem{haupt:10_submitted}
J.~Haupt, R.~Castro, and R.~Nowak, ``Distilled sensing: Adaptive sampling for
  sparse detection and estimation,'' \emph{{IEEE} Trans. Inf. Theory}, vol.~57,
  no.~9, pp. 6222--6235, Sep. 2011.

\bibitem{haupt:08}
------, ``Adaptive discovery of sparse signals in noise,'' in \emph{42nd
  Asilomar Conf. on Signals, Systems, and Computers}, Pacific Grove, CA, Oct.
  2008.

\bibitem{haupt:09}
------, ``Distilled sensing: selective sampling for sparse signal recovery,''
  in \emph{12th Conference on Artificial Intelligence and Statistics},
  Clearwater Beach, FL, Apr. 2009.

\bibitem{David:book}
H.~A. David and H.~N. Nagaraja, \emph{Order Statistics}, 3rd~ed.\hskip 1em plus
  0.5em minus 0.4em\relax New York: Wiley, 2003.

\end{thebibliography}

\begin{IEEEbiographynophoto}{Ali Tajer}(S'05, M'10) is a Postdoctoral Research Associate at Princeton University and an Adjunct Assistant Professor at Columbia University. He received the Ph.D. degree in Electrical Engineering from Columbia University, New York, NY in 2010. Prior to that he received the B.Sc. and M.Sc. degrees in Electrical Engineering from Sharif University of Technology, Tehran, Iran, in 2002 and 2004, respectively. He spent the summers of 2008 and 2009 as a Research Assistant at NEC Laboratories America, Princeton, NJ. His research interests lie in the general areas of network information theory and statistical signal processing.
\end{IEEEbiographynophoto}
\begin{IEEEbiographynophoto}{Rui M. Castro} received a Licenciatura degree in aerospace engineering in 1998 from the Instituto Superior Tecnico, Technical University of Lisbon, Portugal, and a Ph.D. degree in electrical and computer engineering from Rice University, Houston, Texas in 2008.  Between 1998 and 2000, he was a researcher with the Communication Theory and Pattern Recognition Group, Institute of Telecommunications, Lisbon, and in 2002 he held a visiting researcher position at the Mathematics Research Center, Bell Laboratories Research.  He was a postdoctoral fellow at the University of Wisconsin in 2007-2008, and between 2008 and 2010 he held an Assistant Professor position in the department of electrical engineering at Columbia University.  He is currently an assistant professor in the Department of Mathematics and Computer Science of the Eindhoven University of Technology (TU/e), in the Netherlands.  His broad research interests include learning theory, statistical signal and image processing, network inference, and pattern recognition.  Mr. Castro received a Rice University Graduate Fellowship in 2000 and a Graduate Student Mentor Award from the University of Wisconsin in 2008.

\end{IEEEbiographynophoto}
\begin{IEEEbiographynophoto}{Xiaodong Wang} (S'98-M'98-SM'04-F'08) received the Ph.D. degree in Electrical Engineering from Princeton
University. He is a Professor of  Electrical Engineering at Columbia University in New York.
Dr. Wang's research interests fall in the general areas of computing, signal processing
 and communications,
 and has published extensively in these areas. Among his
publications is a book entitled ``Wireless Communication Systems: Advanced Techniques for Signal
Reception'', published by Prentice Hall in 2003.  His current research interests include wireless
communications,   statistical signal processing,
 and genomic signal processing.
  Dr. Wang received the 1999 NSF CAREER Award,   the 2001
 IEEE Communications Society and Information Theory Society Joint Paper Award,
 and the 2011 IEEE Communication Society Award for Outstanding Paper on
 New Communication Topics. He has served  as an
 Associate Editor for the {\em IEEE Transactions on Communications}, the {\em IEEE Transactions on
 Wireless Communications},
   the {\em IEEE Transactions on Signal Processing},
   and   the {\em IEEE Transactions on Information Theory}. He is a Fellow of the IEEE and listed
   as an ISI Highly-cited Author.
\end{IEEEbiographynophoto}
\end{document}